\documentclass[10pt,twocolumn,twoside]{IEEEtran}

\IEEEoverridecommandlockouts        
\usepackage{hyperref}

\title{\LARGE \bf Metrics and Algorithms for Controllability of
  Complex Networks}

\title{\LARGE \bf Controllability Metrics, Limitations and Algorithms
  \\ for Complex Networks}

\author{Fabio Pasqualetti, Sandro Zampieri, and Francesco Bullo %
  \thanks{This material is based upon work supported in part by NSF Award 1135819
and in part by ARO Award W911NF-11-1-0092 and by European
Communitys Seventh Framework Program under grant agreement n. 257462
HYCON2 Network of Excellence.}
  \thanks{Fabio Pasqualetti is with the Mechanical Engineering
    Department, University of California at Riverside,
    \href{mailto:fabiopas@engr.ucr.edu}{\texttt{fabiopas@engr.ucr.edu}}
    .}
  \thanks{Sandro Zampieri is with the Department of Information
    Engineering, University of Padova, Italy,
    \href{mailto:sandro.zampieri@dei.unipd.it}{\texttt{sandro.zampieri@dei.unipd.it}}.}
  \thanks{Francesco Bullo is with the Mechanical Engineering
    Department and the Center for Control, Dynamical Systems and
    Computation, University of California at Santa Barbara,
    \href{mailto:bullo@engineering.ucsb.edu}{\texttt{bullo@engineering.ucsb.edu}}
    .}  }

\usepackage{graphicx,color}
\usepackage{amsmath}
\usepackage{amssymb}
\usepackage{mathrsfs}
\usepackage[vlined,ruled]{algorithm2e}
\usepackage{subfigure}
\usepackage{url}

\usepackage{booktabs}
\usepackage{array}
\usepackage[table]{xcolor}

\newtheorem{theorem}{Theorem}[section]

\newtheorem{corollary}[theorem]{Corollary}

\newtheorem{remark}{Remark}

\newtheorem{example}{Example}

\newcommand{\map}[3]{#1: #2 \rightarrow #3}
\newcommand{\setdef}[2]{\{#1 \; : \; #2\}}
\newcommand{\subscr}[2]{{#1}_{\textup{#2}}}

\newcommand{\until}[1]{\{1,\dots,#1\}}

\newcommand{\real}{\mathbb{R}}

\newcommand{\transpose}{\mathsf{T}} 

\newcommand{\mc}{\mathcal}

\DeclareSymbolFont{bbold}{U}{bbold}{m}{n}
\DeclareSymbolFontAlphabet{\mathbbold}{bbold}

\newcommand\oprocendsymbol{\hbox{$\square$}}
\newcommand\oprocend{\relax\ifmmode\else\unskip\hfill\fi\oprocendsymbol}

\renewcommand{\theenumi}{(\roman{enumi})}

\begin{document}
\maketitle

\thispagestyle{empty}
\pagestyle{empty}

\begin{abstract}
  This paper studies the problem of controlling complex networks, that
  is, the joint problem of selecting a set of control nodes and of
  designing a control input to steer a network to a target state. For
  this problem (i) we propose a metric to quantify the difficulty of
  the control problem as a function of the required control energy,
  (ii) we derive bounds based on the system dynamics (network topology
  and weights) to characterize the tradeoff between the control energy
  and the number of control nodes, and (iii) we propose an open-loop
  control strategy with performance guarantees. In our strategy we
  select control nodes by relying on network partitioning, and we
  design the control input by leveraging optimal and distributed
  control techniques. Our findings show several control limitations
  and properties. For instance, for Schur stable and symmetric
  networks: (i) if the number of control nodes is constant, then the
  control energy increases exponentially with the number of network
  nodes, (ii) if the number of control nodes is a fixed fraction of
  the network nodes, then certain networks can be controlled with
  constant energy independently of the network dimension, and (iii)
  clustered networks may be easier to control because, for
  sufficiently many control nodes, the control energy depends only on
  the controllability properties of the clusters and on their coupling
  strength. We validate our results with examples from power networks,
  social networks, and epidemics spreading.
\end{abstract}

\section{Introduction}
Networks accomplish complex behaviors via local interactions of simple
units. Electrical power grids, mass transportation systems, and
cellular networks are instances of modern technological networks,
while metabolic and brain networks are biological examples. The
ability to control and reconfigure complex networks via external
controls is fundamental to guarantee reliable and efficient network
functionalities. Despite important advances in the theory of control
of dynamical systems, several questions regarding the control of
complex networks are largely unexplored including, for instance, the
relation between network topology and its degree of controllability.

The control problem of complex networks consists of the selection of a
set of control nodes, and the design of a control law to steer the
network to a target state. Inspired by classic controllability notions
for dynamical systems \cite{TK:80,REK-YCH-SKN:63,RWB:70,KJR:88}, we
adopt the worst-case energy to drive a network from the origin to a
target state as controllability metric. By combining this
controllability notion with graph theory, we characterize tradeoffs
between the energy to control a network (with first-order dynamics)
and the number of control nodes, and develop an open-loop distributed
control strategy with guaranteed performance and computational
complexity.


\noindent
\textbf{Related work} The notion of controllability of a dynamical
system was first introduced in \cite{REK-YCH-SKN:63}, and it refers to the
possibility of driving the state of a dynamical system to a specific
target state by means of a control input. Several structural
conditions ensuring controllability have been proposed; see for
instance \cite{TK:80,RWB:70,KJR:88}. The concept of controllability has
received recent interest in the context of complex networks, where
classic methods are often inapplicable due to the system dimension,
and where a graph-inspired understanding of controllability rather
than a matrix-theoretical one is preferable.

Controllability of complex networks is addressed in
\cite{YYL-JJS-ALB:11} by means of graph-theoretic tools from
structured control theory \cite{KJR:88}. In \cite{YYL-JJS-ALB:11} the
application of standard control results to real networks reveals that
the number of control nodes is mainly related to the network degree
distribution, and that sparse inhomogeneous networks are most
difficult to control, while dense and homogeneous networks require
only a few control nodes. Analogous results are derived in
\cite{YYL-JJS-ALB:13} for observability of complex networks. The
approach to controllability and observability undertaken in
\cite{YYL-JJS-ALB:11,YYL-JJS-ALB:13} has several shortcomings. First,
the presented results are \emph{generic}, in the sense that they hold
for \emph{almost every} choice of the network parameters
\cite{WMW:85}, but they may fail to hold if certain symmetries or
constraints are present \cite[Section 15]{KJR:88},
\cite{GP-GN:10}. Second, most results in
\cite{YYL-JJS-ALB:11,YYL-JJS-ALB:13} rely on particular
interconnection properties of the considered networks, perhaps the
absence of self-loops around the network nodes. In fact, it follows
from \cite[Theorem 14.2]{KJR:88}, equivalently from \cite[Theorem
1]{JMD-CC-JW:03}, that every strongly connected network with
self-loops is generically controllable by any single node, which
contradicts the conclusions drawn in \cite{YYL-JJS-ALB:11}. This
discrepancy is underlined in \cite{FJM-AS:11} for the case of
biological networks, and more generally in
\cite{NJC-EJC-DAV-JSF-CTB:12}. Third, the binary notion of
controllability proposed in \cite{REK-YCH-SKN:63} and adopted in
\cite{YYL-JJS-ALB:11} does not characterize the difficulty of the
control task. In practice, although a network may be generically
controllable by any single node, the actual control input may not be
implementable due to actuator constraints and limitations. Finally,
the design of the actual control input to drive a network to a
particular state is not specified in \cite{YYL-JJS-ALB:11}, and it
remains to date an outstanding problem for complex networks, due to
their dimension and absence of a central controller.

We depart from
\cite{YYL-JJS-ALB:11,YYL-JJS-ALB:13,GP-GN:10,NJC-EJC-DAV-JSF-CTB:12},
and analogously from \cite{AR-MJ-MM-ME:09,AO:13,GN-GP:13}, by adopting
a quantitative measure of network controllability, namely the
worst-case control energy, by characterizing tradeoffs between the
difficulty of the control task and the number of control nodes and,
finally, by proposing an open-loop control strategy suitable for
complex networks.

A quantitative approach to network controllability has recently been
adopted in
\cite{GY-JR-YCL-CHL-BL:12,JS-AEM:13,THS-JL:13,AC-MM:12}. With respect
to \cite{GY-JR-YCL-CHL-BL:12}, although our measures of network
controllability coincide, we focus on the tradeoffs between control
energy and number of control nodes, and on the design of a distributed
control strategy, as opposed to scaling laws for the control energy as
a function of the control horizon. With respect to \cite{JS-AEM:13} we
provide a rigorous framework for network controllability and, in fact,
our findings are aligned and mathematically support the discussions in
\cite{JS-AEM:13}. With respect to \cite{THS-JL:13} we adopt a
different network controllability measure, which we show to be more
appropriate for the control of most complex networks. Finally, with
respect to \cite{AC-MM:12}, we consider a more general class of
network dynamics, interconnection graphs, and bounds.


\noindent
\textbf{Paper contributions} The main contributions of this paper are
threefold. First, we study network controllability from an energy
perspective, which we quantify with the smallest eigenvalue of the
controllability Gramian (Section \ref{sec: preliminary}). We show that, if
the number of control nodes is constant, then certain controllable networks
are practically uncontrollable, as the control energy depends exponentially
on the ratio between the network cardinality and the number of control
nodes.


Second, we characterize a tradeoff between the control energy and the
number of control nodes (Section \ref{sec: bounds}). In particular, we
derive an upper bound for the smallest eigenvalue of the
controllability Gramian as a function of the number of control nodes,
and a lower bound on the number of control nodes when the control
energy is fixed.  Our bounds show for instance that the control of
stable and symmetric networks with constant energy requires the number
of control nodes to grow linearly with the network dimension.
These results provide a quantitative measure of the numerical findings
in \cite{JS-AEM:13}, and are in accordance with existing results in
control theory \cite{ACA-DCS-YZ:02}.

Third, we propose the \emph{decoupled control strategy} for the
control of stable complex networks (Section \ref{sec: control}). The
decoupled control strategy consists of network partitioning, selection
of the control nodes, and the design of an open-loop distributed
control law to steer the network from the origin to a target state. We
characterize the performance of the decoupled control strategy and we
show that, with sufficiently many control nodes, the energy to control
a network depends only on the controllability properties of its parts,
and on their coupling strength. Conversely, we prove that certain
networks admit a distributed control strategy where the control energy
is independent of the network dimension. Our decoupled control
strategy constitutes a first scalable open-loop solution for the
distributed control of complex networks, and it leads to a novel
network centrality notion inspired by systems controllability.


Finally, we compare the effectiveness of our decoupled control law
with other network control methods through examples from power
networks, social networks, and epidemics (Section
\ref{sec:simulations}). Our numerical studies show that our decoupled
control strategy outperforms existing control techniques while being
scalable, and amenable to distributed implementation.

Our bounds and techniques apply to diagonalizable networks, and are
simpler and tighter for normal networks, that is, networks with normal
weighted adjacency matrix \cite{RAH-CRJ:85}.

This paper contains three additional minor contributions. First, we
show that the problem of selecting control nodes to maximize the trace
of the controllability Gramian admits a closed-form solution
(Appendix). Second, we generalize our results to the observability
problem of complex networks (Remark \ref{remark: observability}).
Third, we describe a heuristic strategy based on \emph{modal
  controllability} \cite{AMAH-AHN:89} to select control nodes
(Remark~\ref{remark:selection within cluster}).

\noindent
\textbf{Notation} The following notation is adopted throughout the
paper. For a vector $v \in \real^n$, we let $\| v \|_2$ denote its
Euclidean norm, that is,
\begin{align*}
  \| v \|_2 := \sqrt{v^\transpose v},
\end{align*}
where $^\transpose$ denotes transposition. For a matrix $M \in
\real^{n \times n}$, let $\text{spec} (M)$ denote the set of
eigenvalues of $M$, and let
\begin{align*}
  \lambda_\text{min} (M) &:= \min \setdef{|\lambda |}{ \lambda \in
    \text{spec} (M)} ,\\
  \lambda_\text{max} (M) &:= \max \setdef{|\lambda |}{ \lambda \in
    \text{spec} (M)} .
\end{align*}
Let $\sigma (M)$ be the set of the singular values of $M$, that is,
\begin{align*}
  \sigma (M) := \setdef{\lambda^{1/2}}{\lambda \in \text{spec}
    (M^\transpose M)} .
\end{align*}
Let $\sigma_\text{max} (M) := \max \setdef{\lambda}{\lambda \in
  \sigma(M)}$. The spectral norm of $M$ is denoted by $\| M \|_2$,
where
\begin{align*}
  \| M \|_2 := \sigma_\text{max} (M).
\end{align*}
For the vector valued signal $\map{s}{\mathbb{N}_{\ge 0}}{\real^n}$,
we use $\| s \|_{2,T} $ to denote its norm, that is,
\begin{align*}
  \| s \|_{2,T} := \sqrt{\sum_{t = 0}^{T-1} \| s(t) \|_2 }.
\end{align*}
Vector norms, matrix norms, and signal norms will be distinguished
from the context.

\section{Network Model and Preliminary Results}\label{sec:
  preliminary}
Consider a network represented by the directed graph $\mc G := (\mc V,
\mc E)$, where $\mc V := \until{n}$ and $\mc E \subseteq \mc V \times
\mc V$ are the vertices and the edges sets, respectively. Let $a_{ij}
\in \real$ be the weight associated with the edge $(i,j) \in \mc E$,
and define the \emph{weighted adjacency matrix} of $\mc G$ as $A =
[a_{ij}]$, where $a_{ij} = 0$ whenever $(i,j) \not\in \mc E$. We
assume the matrix $A$ to be diagonalizable, that is, $A$ admits a
basis of eigenvectors \cite{RAH-CRJ:85}. We associate a real value
(\emph{state}) with each node, collect the nodes states into a vector
(\emph{network state}), and define the map $\map{x}{\mathbb{N}_{\ge
    0}}{\real^n}$ to describe the evolution (\emph{network dynamics})
of the network state over time. We consider the discrete time, linear,
and time-invariant network dynamics described by the equation
\begin{align}\label{eq:system}
  x(t+ 1) = A x(t).
\end{align}

Controllability of the network $\mc G$ refers to the possibility of
steering the network state to an arbitrary configuration by means of
external controls. We assume that a set 
\begin{align*}
  \mc K := \{k_1, \dots, k_m \}\subseteq \mc V
\end{align*}
of nodes can be independently controlled, and we let 
\begin{align}\label{eq: B}
  B_{\mc K} := 
  \begin{bmatrix}
    e_{k_1} & \cdots & e_{k_m}
  \end{bmatrix}
\end{align}
be the \emph{input matrix}, where $e_i$ denotes the $i$-th canonical
vector of dimension $n$. The network with control nodes $\mc K$ reads
as 
\begin{align}\label{eq:controlled}
  x (t +1) &=  A x(t) + B_{\mc K} u_{\mc K} (t),
\end{align}
where $\map{u_{\mc K}}{\mathbb{N}_{\ge 0}}{\real}$ is the control
signal injected into the network via the nodes $\mc K$. A network is
controllable in $T \in \mathbb{N}$ steps by the set of control nodes
$\mc K$ if and only if for every state $\subscr{x}{f} \in \real^n$
there exists an input $u_{\mc K}$ such that $x(T) = \subscr{x}{f}$
with $x(0) = 0$ \cite{TK:80}. Controllability of dynamical systems is
a well-understood property, and it can be ensured by different
structural conditions \cite{REK-YCH-SKN:63,RWB:70,KJR:88}. For
instance, let $\mc C_{\mc K,T}$, with $T \in \mathbb{N}_{\ge 1}$, be
the \emph{controllability matrix} defined as
\begin{align*}
  \mc C_{\mc K,T} :=
  \begin{bmatrix}
    B_{\mc K} & A B_{\mc K} & \cdots & A^{T-1} B_{\mc K}
  \end{bmatrix}
  .
\end{align*}
The network \eqref{eq:controlled} is controllable in $T$ steps by the
nodes $\mc K$ if and only if the controllability matrix $\mc C_{\mc
  K,T}$ is of full row rank.

The above notion of controllability is qualitative, and it does not
quantify the difficulty of the control task as measured, for instance,
by the control energy needed to reach a desired state. As a matter of
fact, many controllable networks require very large control energy to
reach certain states \cite{JS-AEM:13}. To formalize this discussion,
define the $T$-steps \emph{controllability Gramian} by
\begin{align*}
  \mc W_{\mc K,T} := \sum_{\tau = 0}^{T-1} A^{ \tau} B_{\mc K}
  B_{\mc K}^\transpose (A^\transpose)^{\tau} = \mc C_{\mc K,T} \mc
  C_{\mc K,T}^\transpose, 
\end{align*}
It can be verified that the controllability Gramian $\mc W_{\mc K,T}$
is positive definite if and only if the network is controllable in $T$
steps by the nodes $\mc K$ \cite{TK:80}.

Let the network be controllable in $T$ steps, and let $\subscr{x}{f}$
be the desired final state at time $T$, with $\|\subscr{x}{f} \|_2 =
1$. Define the energy of the control input $u_{\mc K}$ as
\begin{align*}
  \text{E}( u_{\mc K},T) := \| u_{\mc K} \|_{2,T}^2 = \sum_{\tau =
    0}^{T-1} \| u_{\mc K} (\tau) \|_2^2,
\end{align*}
where $T$ is the control horizon. The unique control input that steers
the network state from $x(0) = 0$ to $x(T) = \subscr{x}{f}$ with
minimum energy is \cite{TK:80}
\begin{align}\label{eq:min energy control} 
  u_{\mc K}^* (t) := B_{\mc K}^\transpose (A^\transpose)^{T - t - 1}
  \mc W_{\mc K,T}^{-1} \, \subscr{x}{f} ,
\end{align}
with $t \in \{0,\dots, T-1\}$.
Then, it can be seen that
\begin{align}\label{eq: energy lambda min}
  \text{E}( u_{\mc K}^*, T) = \sum_{\tau = 0}^{T-1} \| u_{\mc K}^* (\tau)
  \|_2^2 = \subscr{x}{f}^\transpose \mathcal{W}_{\mc K,T}^{-1}
  \subscr{x}{f} \le \lambda_\text{min}^{-1} (\mathcal{W}_{\mc K,T}) ,
\end{align}
where equality is achieved whenever $\subscr{x}{f}$ is an eigenvector
of $\mathcal{W}_{\mc K,T}$ associated with $\lambda_\text{min}
(\mathcal{W}_{\mc K,T})$. Because the control energy is limited in
practical applications, controllable networks featuring small Gramian
eigenvalues cannot be steered to certain states.

\begin{example}{\bf \emph{(Controllable networks may exhibit
      practically uncontrollable states)}}\label{example:energy line}
  Consider the network $\mc G$ with $n$ nodes, weighted adjacency
  matrix $A := [a_{ij}]$ defined as
  \begin{align*}
    a_{ij} :=
    \begin{cases}
      \frac{1}{2}, & \text{if } j = i - 1 \text{ and } i \in
      \{2,\dots, n\},\\
      0, & \text{otherwise} ,
    \end{cases}
  \end{align*}
  and control node $\mc K = \{1\}$. Notice that the controllability
  matrix $\mc C_{\mc K,n}$ is diagonal and nonsingular, and its $i$-th
  diagonal entry equals $2^{-i+1}$. Since $A^t B_{\mc K} = 0$ for all
  $t \ge n$, we have $\mc W_{\mc K, \tau} = \mc C_{\mc K,n} \mc C_{\mc
    K,n}^\transpose$ for all $\tau \ge n$, and the smallest eigenvalue
  of the controllability Gramian $\mc W_{\mc K, \tau}$ equals $2^{-2n
    + 2}$ for all $\tau \ge n$. We conclude that the network $\mc G$
  with control node $\mc K$ is controllable in $T \ge n$ steps, yet
  the control energy grows exponentially with the network cardinality.
  \oprocend
\end{example}
\smallskip

In this work we measure controllability of a network based on the
smallest eigenvalue of the controllability Gramian. With this choice
we study controllability from a worst-case perspective, looking at the
target states requiring the largest control energy to be reached; see
also \cite{IR-MG-MM:11}. We conclude this section by discussing
alternative controllability metrics.

\begin{remark}{\bf \emph{(Controllability metrics)}}\label{remark:
    cost functions}
  Different quantitative measures of controllability of dynamical
  systems have been considered in the last years \cite{PCM-HIW:72}. In
  addition to the smallest eigenvalue of the controllability Gramian
  $\lambda_\text{min} (\mathcal{W}_{\mc K,T})$, the trace of the
  inverse of the controllability Gramian $\text{Trace} (\mc W_{\mc
    K,T}^{-1})$, and the determinant of the controllability Gramian
  $\text{Det} (\mc W_{\mc K,T})$ have been proposed. It can be shown
  that, while $\text{Trace} (\mc W_{\mc K,T}^{-1})$ measures the
  average control energy over random target states, $\text{Det} (\mc
  W_{\mc K,T})$ is proportional to the volume of the ellipsoid
  containing the states that can be reached with a unit-energy control
  input. The selection of the control nodes for the optimization of
  these metrics is usually a computationally hard combinatorial
  problem \cite{AO:13}, for which heuristics without performance
  guarantees and non-scalable optimization procedures have been
  proposed \cite{AMAH-AHN:89,MVDW-BDJ:01,LSP-KRK:99}.

  Motivated by the relation
  \begin{align*}
    \frac{\text{Trace} (\mc W_{\mc K,T}^{-1})}{n} \ge
    \frac{n}{\text{Trace} (\mc W_{\mc K,T})} ,
  \end{align*}
  the trace of the controllability Gramian $\text{Trace} (\mc W_{\mc
    K,T})$ has also been used as an overall measure of controllability
  in \cite{BM-DK-DG:04,HRS-MT:12}, and recently in
  \cite{THS-JL:13}. Unlike the controllability metrics
  $\lambda_\text{min} (\mathcal{W}_{\mc K,T})$, $\text{Trace} (\mc
  W_{\mc K,T}^{-1})$, and $\text{Det} (\mc W_{\mc K,T})$, the
  selection of the control nodes to maximize $\text{Trace} (\mc W_{\mc
    K,T})$ admits a closed-form solution (see
  Appendix). Unfortunately, the maximization of $\text{Trace} (\mc
  W_{\mc K,T})$ does not automatically ensure controllability and, as
  we show in Sections \ref{section:example circle} and
  \ref{sec:simulations}, it often leads to a poor selection of the
  control nodes with respect to the worst-case control energy to reach
  a target state.  \oprocend
\end{remark}

\section{Control Nodes and Control Energy}\label{sec: bounds}
In this section we characterize a tradeoff between the number of
control nodes and the energy required to drive a network to a target
state. Recall that the condition number of an invertible matrix $M$ is
$\text{cond} (M) = \| M \|_2 \| M^{-1}\|_2$.

\begin{theorem}{\bf \emph{(Control energy and number of control nodes
      for unstable networks)}}\label{thm:upper bound energy unstable
    networks}
  Consider a network $\mc G = (\mc V, \mc E)$ with $|\mc V| = n$,
  weighted adjacency matrix $A$, and control set $\mc K$. Assume that
  $A$ is diagonalizable by the eigenvector matrix $V$, and let
  $\lambda_\text{min} (A) < 1$. Let $\mu \in \real_{\ge 0}$, and let
  \begin{align*}
    n_\mu = \left| \setdef{\lambda}{\lambda \in \text{spec} (A),
        |\lambda | \le \mu} \right|.
  \end{align*}
  For all $T \in \mathbb{N}_{>0}$ and for all $\mu \in [
  \lambda_\text{min} (A),1)$ it holds
  \begin{align*}
    \lambda_\text{min} (\mc W_{\mc K,T}) \le \text{cond}^2 (V)\frac{ \mu^{2
        \left(\left\lceil \frac{n_\mu }{|\mc K|} \right\rceil - 1
        \right)} }{1 - \mu^2 } .
  \end{align*}
\end{theorem}
\begin{proof}
  Let $V$ be an eigenvector matrix for $A$, and assume that the
  columns of $V$ are ordered such that
  \begin{align*}
    V^{-1} A V =
    \begin{bmatrix}
      A_1 & 0\\
      0 & A_2
    \end{bmatrix}, \; 
    V^{-1} B_{\mc K} =
    \begin{bmatrix}
      B_1 \\
      B_2
    \end{bmatrix},
  \end{align*}
  where $A_1 \in \real^{n_\mu \times n_\mu}$ and $A_2$ are diagonal
  matrices, and $\text{spec} (A_1) = \setdef{\lambda}{\lambda \in
    \text{spec} (A), |\lambda | \le \mu}$.  Observe that
  \begin{align*}
    \mc W_{\mc K, T} &= \sum_{\tau = 0}^{T-1} A^\tau B_{\mc K} B_{\mc
      K}^\transpose (A^\transpose)^\tau \\ 
    &= V  \underbrace{\sum_{\tau = 0}^{T-1} 
        \begin{bmatrix}
      A_1 & 0\\
      0 & A_2
    \end{bmatrix}^\tau
    \begin{bmatrix}
      B_1 \\ B_2
    \end{bmatrix}
    \begin{bmatrix}
      B_1 \\ B_2
    \end{bmatrix}^\transpose
    \begin{bmatrix}
      A_1 & 0\\
      0 & A_2
    \end{bmatrix}^\tau}_{\tilde W_{\mc K, T}}
    V^\transpose ,
  \end{align*}
  where we have used the fact that $A_1$ and $A_2$ are
  symmetric. Since $\mc W_{\mc K, T}$ is symmetric, we have
  \begin{align*}
    \lambda_\text{min} (\mc W_{\mc K, T}) &= \min_{\| x\| = 1}
    x^\transpose \mc W_{\mc K, T} x = \min_{\| x\| = 1} x^\transpose V
    \tilde{\mc W}_{\mc K, T} V^\transpose x \\
    &\le \| V\|^2_2 \min_{\| y\|=1} y^\transpose \tilde{\mc W}_{\mc K,
      T} y \\ &\le \| V\|_2^2 \min_{\| y_1\|=1} 
    \begin{bmatrix}
      y_1^\transpose & 0
    \end{bmatrix}
    \tilde{\mc W}_{\mc K, T} 
    \begin{bmatrix}
      y_1 \\ 0
    \end{bmatrix}
    \\ &= \| V\|^2_2 \lambda_\text{min} \left( \sum_{\tau = 0}^{T-1} A_1^\tau
      B_1 B_1^\transpose A_1^\tau \right) .
  \end{align*}
  Let $T_\text{max} = \left\lceil \frac{n_\mu}{|\mc K|} \right\rceil
  -1$, and notice that the matrix
  \begin{align*}
    \sum_{\tau= 0}^{T_\text{max}-1 } A_1^\tau B_1 B_1^\transpose
    A_1^\tau = \mc C_{1,T_\text{max}} \mc
    C_{1,T_\text{max}}^\transpose
  \end{align*}
  is singular, where $\mc C_{1,T_\text{max}}$ is the controllability
  matrix of $(A_1,B_1)$ at $T_\text{max}$ steps. In fact, $\mc
  C^1_{\mc K , T_\text{max}} \in \real^{n_\mu \times m}$ with
  \begin{align*}
    m = T_\text{max} |\mc K| < \left( \frac{n_\mu}{|\mc K|} + 1 \right)
    |\mc K| -|\mc K| = n_\mu .
  \end{align*}
  An application of the Bauer-Fike theorem
  \cite{RAH-CRJ:85,FLB-CTF:60} for the location of eigenvalues of
  perturbed matrices yields
  \begin{align*}
    &\lambda_\text{min} (\mc W_{\mc K, T}) \le \| V\|_2^2 \left(
      \lambda_\text{min} \left( \sum_{\tau= 0}^{T-1 }
        A_1^\tau B_1 B_1^\transpose A_1^\tau \right) \right) \\
    &\le \| V\|_2^2 \Bigg( \lambda_\text{min} \left( \sum_{\tau=
        0}^{T_\text{max}-1 } A_1^\tau B_1 B_1^\transpose A_1^\tau
    \right) \\ &\quad + \left\| \sum_{\tau= T_\text{max}}^{T-1 } A_1^\tau
      B_1 B_1^\transpose
      A_1^\tau \right\|_2 \Bigg) \\
    &\le \| V\|_2^2 \sum_{\tau = T_\text{max}}^{T-1} \| A_1 \|^{2 \tau}_2
    \| B_1 \|^2_2
    \le \underbrace{\| V\|_2^2 \| V^{-1}\|_2^2}_{\text{cond}^2 (A)}
    \frac{ \mu^{2 \left(\left\lceil \frac{n_\mu }{|\mc K|}
          \right\rceil - 1 \right)} }{1 - \mu^2 },
  \end{align*}
  where we have used the facts that $A_1$ is diagonal,
  $\lambda_\text{max} (A_1) \le \mu$, and $\| B_1\|_2 \le \| V^{-1}
  B_{\mc K}\|_2 \le \| V^{-1}\|_2$.
\end{proof}

In Theorem \ref{thm:upper bound energy unstable networks} we provide
an upper bound on the smallest eigenvalue of the controllability
Gramian or, equivalently, a lower bound on the worst-case energy
needed to control a network to an arbitrary target state, as a
function of the eigenvalues distribution of $A$ and the condition
number of the set of its eigenvectors. The bound in Theorem
\ref{thm:upper bound energy unstable networks} needs to be regarded as
a performance limitation: independently of the control strategy
adopted by the control nodes, the least amount of energy needed to
steer the network to an arbitrary unit-norm state is bounded by the
inverse of the expression in Theorem \ref{thm:upper bound energy
  unstable networks}. Notice that Theorem \ref{thm:upper bound energy
  unstable networks} contains a family of bounds, because equation
\eqref{eq: bound energy} holds for all values $\mu \in
[\lambda_\text{min} (A), 1 )$. Finally, equation \eqref{eq: bound
  energy} simplifies when $A$ is normal (in particular when $A$ is
symmetric) due to the existence of an orthonormal eigenvector matrix
$V$ yielding $\text{cond} (V) = 1$. Indeed, in the case of stable and
symmetric networks simpler and sharper bounds can be obtained as a
corollary of Theorem \ref{thm:upper bound energy unstable
  networks}. Recall that a matrix $M$ is Schur stable if
$\lambda_\text{max} (M) < 1$ \cite{RAH-CRJ:85}.

\begin{corollary}{\bf \emph{(Control energy and number of control
      nodes for stable and symmetric networks)}}\label{thm:upper bound
    energy}
  Consider a network $\mc G = (\mc V, \mc E)$ with $|\mc V| = n$,
  weighted adjacency matrix $A$, and control set $\mc K$. Assume that
  $A$ is Schur stable and symmetric. For all $T \in \mathbb{N}_{>0}$
  it holds
  \begin{align}\label{eq: bound energy}
    \lambda_\text{min} (\mc W_{\mc K,T}) \le \min \left\{ \frac{1 -
        \lambda_\text{min}^{2T} (A) }{1 - \lambda_\text{min}^2 (A) }, 
\frac{ \lambda_\text{max}^{2
        \left( \left\lceil \frac{n}{|\mc K|} \right\rceil -1 \right) }(A)
    }{1 - \lambda_\text{max}^2 (A) } \right\}.
  \end{align}
\end{corollary}
\begin{proof}
  We start by showing the first part of the inequality. Notice that
  $B_{\mc V} = [B_{\mc K} \; B_{\mc V \setminus \mc K}]$, and that
  $\mc W_{\mc V ,T} = \mc W_{\mc K ,T} + \mc W_{\mc V \setminus \mc
    K,T}$. Since both $\mc W_{\mc K ,T}$ and $\mc W_{\mc V \setminus
    \mc K,T}$ are positive semi-definite, we conclude that
  $\lambda_\text{min} (\mc W_{\mc K,T}) \le \lambda_\text{min} (\mc
  W_{\mc V ,T})$. Then,
  \begin{align*}
    \lambda_\text{min} (\mc W_{\mc V ,T}) = \lambda_\text{min}
    \left( \sum_{\tau = 0}^{T-1} A^{2 \tau} \right) = 
    \frac{1 -  \lambda_\text{min} (A)^{2T} }{1 - \lambda_\text{min}
      (A)^2 } ,
  \end{align*}
  where we have used the assumption $A = A^\transpose$. The second
  part of the inequality follows from Theorem \ref{thm:upper bound
    energy unstable networks} with $\mu = \lambda_\text{max} < 1$, and
  the fact that symmetric matrices admit an orthonormal eigenvector
  matrix $V$, so that $\text{cond} (V) = 1$.
\end{proof}

  \begin{figure}
    \centering
    \includegraphics[width=.9\columnwidth]{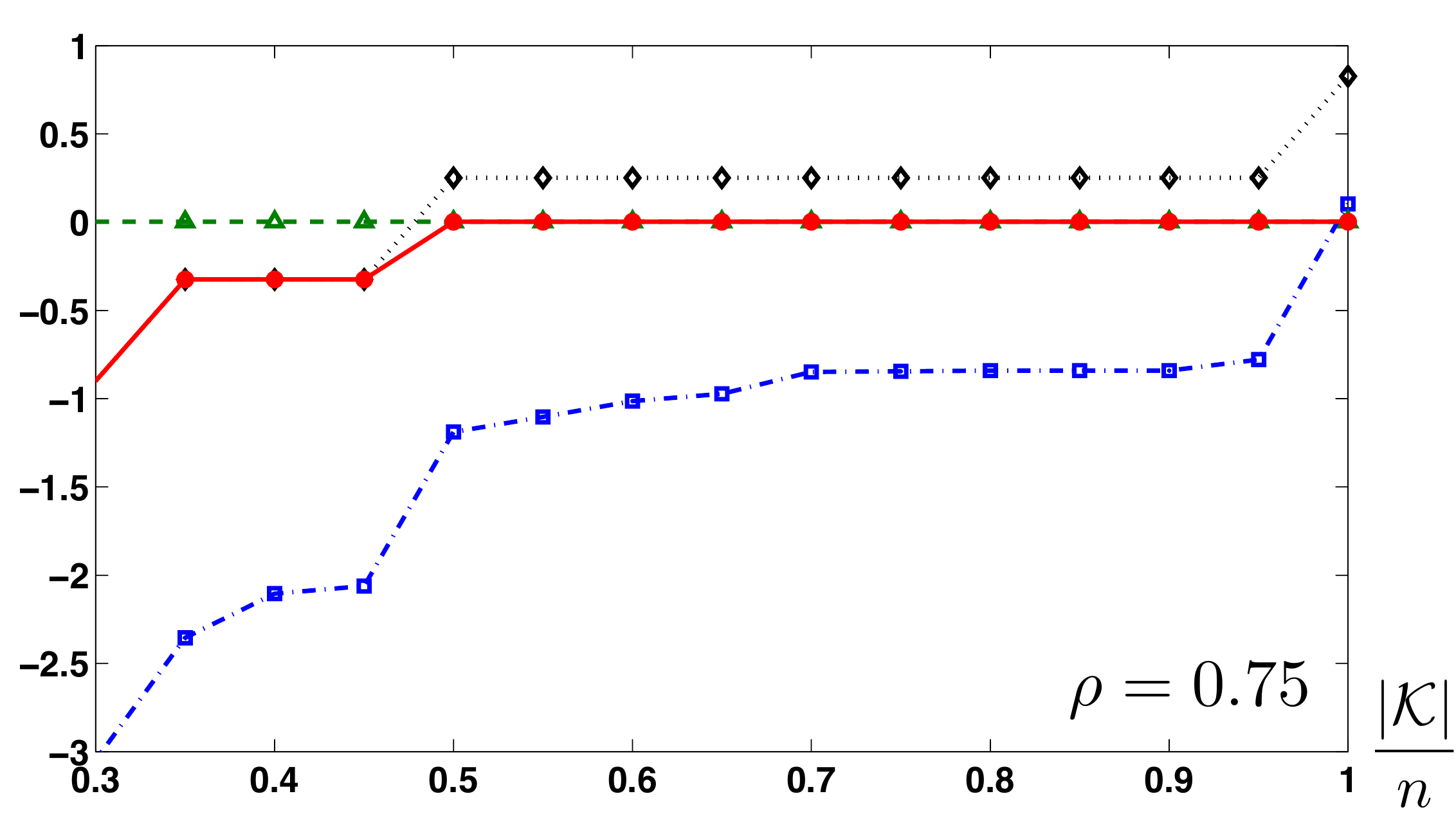}
    \caption{For the network in Example \ref{example: bound tight},
      this figure compares (in a logarithmic scale) the upper bound
      \eqref{eq: bound energy} (solid red) with the largest
      $\lambda_\text{min}$ of the controllability Gramian (dashed-dot
      blue) over all possible sets $\mc K$. For each value $|\mc K|$
      from $1$ to $n$, a combinatorial search determines the value
      $\lambda_\text{min}^* = \max_{\mc K} \lambda_\text{min} (\mc
      W_{\mc K,\infty})$. The two quantities in the right hand side of
      equation \eqref{eq: bound energy} are also reported in dashed
      green and dotted black, respectively. It can be shown that the
      bound \eqref{eq: bound energy} tends to be conservative as the
      $\rho$ increases.}
    \label{fig:UpperBoundCircle}
\end{figure}

\begin{example}{\bf \emph{(Tightness of the bound in Theorem
      \ref{thm:upper bound energy unstable
  networks})}}\label{example: bound tight}
Consider a network with $n = 20$ nodes and adjacency matrix
  \begin{align*}
    A := \frac{\rho}{3}
    \begin{bmatrix}
      1 & 1 & 0 & \cdots & 1\\
      1 & 1 & 1 & \cdots & 0\\
      \vdots & \ddots & \ddots & \ddots & \vdots\\
      0 & \cdots & 1 & 1 & 1\\
      1 & \cdots & 0 & 1 & 1
    \end{bmatrix},
  \end{align*}
  where $\rho \in
  \real_{>0}$. 
  In Fig. \ref{fig:UpperBoundCircle} we select $\rho = 0.75$ and
  compare the upper bound in Corollary \ref{thm:upper bound energy}
  with the value $\max \setdef{\lambda_\text{min} (\mc W_{\mc
      K,T})}{\mc K \subseteq \until{12}, |\mc K| = k}$, as a function
  of the number $k$ of control nodes.  \oprocend
\end{example}

In what follows we consider two asymptotic control scenarios, where
the network cardinality grows, and either the number of control nodes
or the desired control energy remain constant. From Corollary
\ref{thm:upper bound energy} we conclude that for stable and symmetric
networks, if $|\mc K|$ is constant, then the controllability energy
$\lambda_\text{min}^{-1} (\mc W_{\mc K,T})$ grows at least
exponentially as the cardinality $n$ grows.
This reasoning provides a quantitative measure of the findings in
\cite{JS-AEM:13}, and it is in accordance with
\cite{ACA-DCS-YZ:02}. We next consider the case of bounded control
energy.

\begin{corollary}{\bf \emph{(Lower bound on the cardinality of the
      control set)}}\label{corollary: bounded energy}
  Consider a network $\mc G = (\mc V, \mc E)$ with $|\mc V| = n$,
  weighted adjacency matrix $A$, and control set $\mc K$. Let $T \in
  \real_{> 0}$ and $\varepsilon \in \real_{> 0}$.
  If $\lambda_\text{min} (\mc W_{\mc K,T}) \ge \varepsilon$, then
  \begin{align*}
    |\mc K| \ge R_\varepsilon   n_\mu,
  \end{align*}
  where $\mu$, $n_\mu$ and $V$ are as in Theorem \ref{thm:upper bound
    energy unstable networks}, and
  \begin{align*}
    R_\varepsilon := \frac{2 \log( \mu )}{ \log(\varepsilon) + \log (\mu^2
      (1- \mu^2) ) - 2\log( \text{cond} (V) )  
    } .
  \end{align*}
\end{corollary}
\medskip
\begin{proof}
  From Theorem \ref{thm:upper bound energy unstable networks} it
  follows that $\lambda_\text{min} (\mc W_{\mc K,T}) \ge \varepsilon$
  only if
  \begin{align*}
    \varepsilon \left( 1 - \mu^2 \right) \text{cond}^{-2} (V)\le
    \mu^{2 \left( \lceil n_\mu / |\mc K| \rceil - 1 \right)} ,
  \end{align*}
  or, equivalently, only if $|\mc K| \ge R_\varepsilon n_\mu$.
\end{proof}

The previous result has interesting consequences. For instance for
stable and symmetric networks, Corollary \ref{corollary: bounded
  energy} with $\mu = \lambda_\text{max}$ and $\text{cond} (V) = 1$
implies that, in order to guarantee a certain bound on the control
energy, the number of control nodes must be a linear function of the
total number of nodes. Instead, classic controllability
\cite{REK-YCH-SKN:63,YYL-JJS-ALB:11} is (generically) ensured by the
presence of a single control node, independently of the network
dimension \cite[Theorem 14.2]{KJR:88}, \cite[Theorem 1]{JMD-CC-JW:03}.
In fact, Corollary \ref{corollary: bounded energy} can also be used to
show that a similar behavior might appear also for unstable and/or
asymmetric networks. We show this fact through two examples.

\begin{example}{\bf \emph{(Control nodes for circulant
     marginally stable network)}}\label{example: circulant}
 Consider the circulant network in Example \ref{example: bound tight}
 with $\rho = 1$ and $n$ nodes. Let $\mu = 1/3$, and let $\varepsilon$
 be a desired lower bound for the smallest eigenvalue of the
 controllability Gramian. From Corollary \ref{corollary: bounded
   energy} the number of control nodes satisfies
  \begin{align*}
    |\mc K| \ge \frac{\log(3) }{\log(1/\varepsilon) + \log(81/8) }   (n-1),
  \end{align*}
  where we have used that, for $\mu=1/3$, $n_\mu \ge \frac{1}{2} (n -1)$.
  \oprocend
\end{example}

\begin{example}{\bf \emph{(Bound for asymmetric line
      network)}}\label{example: bound asymmetric}
  Consider a network with $n$ nodes and weighted adjacency matrix
    \begin{align*}
    A := \frac{1}{3}
    \begin{bmatrix}
      1 & 1/2 & 0 & 0&\cdots \\
      2 & 1 & 2 &0& \cdots \\
      0& 1/2 & 1 & 1/2 & \cdots \\
      0& 0& 2 & 1  & \ddots \\
      \vdots & \vdots & \vdots & \ddots & \ddots &\\
    \end{bmatrix}.
  \end{align*}
  Define the diagonal matrix $D = \text{diag} (1, 2,1,2,\dots )$.  It
  can be verified that $D^{-1} A D $ is symmetric. Let $\tilde V$ be
  an orthonormal eigenvector matrix of $D^{-1} A D $, and notice that
  $D \tilde V$ is an eigenvector matrix of $A$ with $\text{cond} (D
  \tilde V)=2$. Notice moreover that the eigenvalues of $A$ are
  \begin{align*}
    \lambda_h=\frac{1}{3}\left(1+2\cos\frac{h\pi}{n+1}\right),\qquad
    h\in\{1,\ldots,n\},
  \end{align*}
  so that for $\mu=1/3$ it holds $n_\mu\ge n/2$. Then, Corollary
  \ref{corollary: bounded energy} implies that the number of control
  nodes satisfies
  \begin{align*}
    |\mc K| \ge \frac{\log(3) }{\log(1/\varepsilon) + \log(81/2) }   n.
  \end{align*}
  \oprocend
\end{example}

We remark that the technique used in the previous examples can be used
to determine bounds on the number of control nodes in more general
unstable and marginally stable networks with known eigenvalues
distribution, such as the case of consensus dynamics over random
geometric networks \cite{SB-AG-BP-DS:06}.

\begin{remark}{\bf \emph{(Observability of Complex
      Networks)}}\label{remark: observability}
  The observability problem of complex networks consists of selecting
  a set of sensor nodes, and designing an estimation strategy to
  reconstruct the network state from measurements collected by the
  sensor nodes \cite{YYL-JJS-ALB:13}. Our quantitative analysis of the
  controllability of complex networks in Section \ref{sec: bounds},
  and our decoupled control strategy in Section \ref{sec: control} can
  be directly applied to the problem of observability of complex
  networks. To see this, define the \emph{$T$-steps observability
    Gramian} by
\begin{align*}
  \mc O_{\mc K , T} := \sum_{\tau = 0}^{T-1}(A^\transpose)^\tau C_{\mc
    K}^\transpose C_{\mc K} A^\tau,
\end{align*}
where $\mc K$ denotes the set of sensor nodes, and $C_{\mc K} :=
B_{\mc K}^\transpose$. The energy associated with the network state
$x$ with sensor nodes $\mc K$ and observation horizon $T$ is
\begin{align*}
  \text{E} (x,T) := \sum_{\tau = 0}^{T-1} \| y_{\mc K}(\tau) \|_2^2 =
  x^\transpose \mc O_{\mc K , T} x \ge \lambda_\text{min} (\mc O_{\mc
    K , T}),
\end{align*}
where $\map{y_{\mc K}}{\mathbb{N}_{\ge 0}}{\real}$ contains the
measurements taken by the observing nodes $\mc K$ \cite{DG:95}. Thus,
the smallest eigenvalue of the observability Gramian is a suitable
metric to measure observability of a network. 
The results in Section \ref{sec: bounds} are readily applicable to the
network observability problem. For instance, from Corollary
\ref{thm:upper bound energy} we conclude that the smallest eigenvalue
of the observability Gramian of a stable and symmetric network
decreases exponentially as the ratio of the network cardinality and
the number of sensor nodes grows. \oprocend
\end{remark}

\section{Decoupled Control of Complex Networks}\label{sec: control}
In this section we provide a solution to the problem of controlling a
complex network, that is, the problems of both selecting the control
nodes, and designing a distributed control law to drive the network to
a target state. Our approach is different from classic solutions, as
it exploits the network structure to jointly select the control nodes
and to design an open-loop control law amenable to distributed
implementation.

The problem of selecting control nodes in a dynamical system to
optimize a controllability metric is a classic control problem
\cite{MVDW-BDJ:01}. Most existing solutions either rely on
combinatorial or non-scalable optimization techniques, being therefore
not suited for large networks \cite{LSP-KRK:99}, or are heuristic, in
that they exploit the specific structure of the system at hand, and do
not offer guarantees on the control energy
\cite{AMAH-AHN:89,MVDW-BDJ:01,SSR-TSP-VBV:91,KBL:92}. See Remark
\ref{remark:selection within cluster} for a heuristic method to select
control nodes.


\subsection{Setup and definition of the decoupled control
  strategy}\label{sec: setup decoupling law}
Our open-loop \emph{decoupled control strategy} can be divided into
three parts: (i) network partitioning, (ii) selection of the control
nodes, and (iii) definition of the decoupled control law.

\noindent
\textbf{Network partitioning} Consider an undirected network $\mc G :=
(\mc V, \mc E)$ with weighted adjacency matrix $A :=
[a_{ij}]$. Partition $\mc V$ into $N$ disjoint sets $\mc P :=\{\mc
V_1,\dots, \mc V_N\}$, and let $\mc G_{i} := (\mc V_i ,
\mathcal{E}_i)$ be the $i$-th subgraph of $\mc G$ with vertices $\mc
V_i$ and edges $\mathcal{E}_i := \mathcal{E} \cap (\mc V_i \times \mc
V_i)$.\footnote{Several methods are available to partition a network
  \cite{SF:10}. For the implementation of our decoupled control law it
  is only required that the network is partitioned into strongly
  connected components. The performance of the decoupled control law
  depend on the partitioning scheme, and it remains an outstanding
  problem to design an optimal partitioning algorithm for the
  implementation of the proposed control law. In Section \ref{sec:
    nodes selection} we employ a spectral method based on the Fiedler
  eigenvector to partition a network.} According to this partition,
and possibly after relabeling states and inputs, the network matrices
read as
\begin{align}\label{eq: decomposition}
  A =
  \begin{bmatrix}
    A_{1} & \cdots & A_{1N}\\
    \vdots & \vdots & \vdots\\
    A_{N1} & \cdots & A_{N}
  \end{bmatrix}
  ,\;
  B_{\mc K} =
  \begin{bmatrix}
    B_{\mc K_1} & \cdots & 0\\
    \vdots & \ddots & \vdots\\
    0 & \cdots & B_{\mc K_N}
  \end{bmatrix}
  ,
\end{align}
where $\mc K_i \subseteq \mc V_i$ for all $i \in \until{N}$, and the
networks dynamics can be written as the interconnection of $N$
subsystems of the form
\begin{align}\label{eq: coupled system}
  x_i (t+1) = A_i x_i(t) + \sum_{j \in \mc N_i} A_{ij} x_j (t) +
  B_{\mc K_i} u_{\mc K_i} (t) ,
\end{align}
where $i \in \until{N}$ and $\mc N_i := \setdef{j}{ A_{ij} \neq 0}$.


\noindent
\textbf{Selection of the control nodes} For a network $\mc G := (\mc
V, \mc E)$ with partition $\mc P := \{ \mc V_1, \dots, \mc V_N \}$, we
say that a node $i \in \mc V_k$ is a \emph{boundary} node if $a_{ij}
\neq 0$ for some node $j \in \mc V_\ell$, with $k,\ell \in \until{N}$
and $k \neq \ell$. Let $\Psi_i \subseteq \mc V_i$ be the set of
boundary nodes of the $i$-th cluster, and let $\Psi = \bigcup_{i =
  1}^N \Psi_i$ be the set of all the boundary nodes of the partition
$\mc P$. We select the set of control nodes $\mc K = \mc K_1\cup
\cdots \cup \mc K_N$ to satisfy $\Psi_i \subseteq \mc K_i \subseteq
\mc V_i$ for all $i \in \until{N}$, and so that each pair $(A_i,B_i)$
is controllable. See Fig.  \ref{fig:circle} for an example.

We remark that the set of boundary nodes may not be sufficient to
guarantee controllability of each pair $(A_i, B_i)$. However, if each
cluster is connected and every diagonal entry of the network matrix is
nonzero then, due to genericity of the controllability property
\cite[Theorem 14.2]{KJR:88}, \cite[Theorem 1]{JMD-CC-JW:03}, each
cluster and the whole network are generically controllable by the
boundary nodes. For networks where generic controllability is not
sufficient, existing graph-theoretical algorithms can be used to
select extra control nodes to ensure controllability
\cite{KJR:88,YYL-JJS-ALB:11,AR-MJ-MM-ME:09}.

\noindent \textbf{The decoupled control law} For a network $\mc G :=
(\mc V, \mc E)$ with partition $\mc P := \{ \mc V_1, \dots, \mc V_N
\}$, let $\subscr{x}{f}^\transpose :=
\begin{bmatrix}
  x_{\text{f} 1}^\transpose &
  \cdots &
  x_{\text{f} N}^\transpose
\end{bmatrix}$ be the target state, where $\| \subscr{x}{f} \|_2 =
1$, 
and $x_{\text{f} i} \in \real^{|\mc V_i|}$ for $i \in \until{N}$.  Let
$\| x_{\text{f},i} \|_2 = \alpha_i $, and notice that $\sum_{i = 1}^N
\alpha_i^2 = 1$. Define the control input $u_{\mc K_i}$ by
\begin{align}\label{eq:decoupled law}
  u_{\mc K_i} (t) &:= \underbrace{B_{\mc K_i}^\transpose
    (A_i^\transpose)^{T - t - 1} \mc W_{i,T}^{-1} x_{\text{f}i}}_{v_i
    (t)} - \sum_{j \in \mc N_i} \underbrace{B_{\mc K_i}^\transpose
    A_{ij} x_j (t)}_{f_{ij} (t)} ,
\end{align}
where, with a slight abuse of notation, $\mc W_{i,T}$ is the $i$-th
controllability Gramian defined by
\begin{align*}
  \mc W_{i,T} := \sum_{\tau = 0}^{T-1} A_i^{T - \tau -1} B_{\mc K_i}
  B_{\mc K_i}^\transpose (A_i^\transpose)^{T - \tau -1},
\end{align*}
and the control horizon $T$ is chosen large enough so that $\mc
W_{i,T}$ is positive definite for all $i \in \until{N}$. We refer to
the above control law as to the \emph{decoupled control law}.

Before analyzing the performance of our decoupled control law we
discuss its implementation properties. First, notice that the control
input $u_{\mc K_i}$ is the sum of an open-loop control signal $v_i$,
and a feedback control signal $\sum_{j \in \mc N_i} f_{ij}$. Second,
if each cluster is equipped with a control center, then our decoupled
control law can be implemented via distributed computation by the
control centers. In fact, the control signal $v_i$ depends on the
dynamics of only the $i$-th cluster, and the feedback control signals
$f_{ij}$ can be determined upon communication of the $i$-th control
center with its neighboring control centers $\mc N_i$. Third, our
decoupled control law is scalable, in the sense that the complexity of
the control law does not depend upon the network cardinality, but only
on its partition. We further discuss this property in Section
\ref{section:example circle} and Section
\ref{sec:simulations}. Finally, the decoupled control strategy relies
on an open-loop mechanism. As such, the application of the decoupled
control strategy to real networks would require the presence of a
feedback mechanism to account for unmodeled dynamics and noise in the
system dynamics and measurements.

\subsection{Analysis of the decoupled control law}\label{sec: analysis
decoupled}
We start our analysis by noticing that the decoupled control law
\eqref{eq:decoupled law} steers the network to the target state
$\subscr{x}{f}$. In fact, from equation \eqref{eq: coupled system} and
the definition of $f_{ij}$ in equation \eqref{eq:decoupled law}, the
network dynamics with decoupled control law can be written as the
collection of $N$ decoupled subsystems
\begin{align}\label{eq:decoupled system}
  x_i (t+1) = A_i x_i (t) + B_{\mc K_i} v_i (t) , \; i \in \until{N} .
\end{align}
Since $v_i$ in equation \eqref{eq:decoupled law} equals the minimum
energy input to drive the $i$-th subsystem \eqref{eq:decoupled system}
from $x_i (0) = 0$ to $x_i (T) = x_{\text{f}i}$, we conclude that
$x(T) = \subscr{x}{f}$.

We next study the energy properties of our decoupled control law.
Observe that the state evolution of the \mbox{$i$-th} cluster can be
written as \begin{align*} x_i (t) = \sum_{\tau = 0}^{t - 1} A_i^{t -
    \tau - 1} B_{\mc K_i} B_{\mc K_i}^\transpose (A_i^\transpose)^{t -
    \tau - 1} \mc W_{i , T}^{-1} x_{\text{f} i} .
\end{align*}

In this work we assume that the matrix $A_i$ is Schur stable for all
$i \in \until{N}$, and we leave the case of unstable networks as the
subject of future investigation.
Observe that, if $A$ is Schur stable and nonnegative, then each matrix
$A_i$ is Schur stable and $\lambda_\text{max}(A_i)\le
\lambda_\text{max}(A)$. We define
the \emph{local energy matrix}
$\Lambda \in \real^{N \times N}$ and the \emph{$\mc L_2$ gains matrix}
$\Gamma \in \real^{N \times N}$ by
\begin{align}\label{eq:gain1}
    \Lambda &:= \text{diag} (\lambda_\text{min}^{-1} (\mc W_{1,T}) , \dots ,
    \lambda_\text{min}^{-1} (\mc W_{N,T})) ,\\[.5em]
    \Gamma &:=
    \begin{bmatrix}
      1 & \gamma_{12} & \cdots & \gamma_{1N} \\
      \gamma_{21} & 1 & \cdots &  \gamma_{2N}\\
      \vdots & \vdots & \ddots & \vdots \\
      \gamma_{N1} & \gamma_{N2} & \ddots & 1
    \end{bmatrix}\label{eq:gain2}
    ,
\end{align}
where $\gamma_{ij}$, for $i,j \in \until{N}$ and $i \neq j$, is the
$\mc L_2$ gain of the input-output system $(A_j, B_{\mc K_j}, B_{\mc
  K_i}^\transpose A_{ij})$ or, equivalently, the $H_\infty$ gain of
the transfer matrix $B_{\mc K_i}^\transpose A_{ij}(zI-A_j)^{-1}B_{\mc
  K_j}$ \cite{SS-IP:05}.

\begin{theorem}{\bf \emph{(Energy of the decoupled control
      law)}}\label{thm: energy decoupled}
  Consider a network $\mc G = (\mc V, \mc E)$ with weighted adjacency
  matrix $A$, control set $\mc K$, and partition $\mc P$. Assume that
  $\mc K$ contains all boundary nodes of $\mc P$, every $A_i$ is
  stable, and every pair $(A_i,B_i)$ is controllable, where $A_i$ and
  $B_i$ are the submatrices associated with the partition $\mc P$. The
  decoupled control law $u_{\mc K}^\text{d}$ with control horizon $T$
  satisfies
  \begin{align}\label{eq:bound-input-dec}
    \text{E} ( u_{\mc K}^\text{d}, T) \le \| \Gamma \Lambda^{1/2} \|_2^2,
  \end{align}
  where $\Lambda$ and $\Gamma$ are the local energy matrix and the
  $\mc L_2$ gains matrix defined in \eqref{eq:gain1} and \eqref{eq:gain2},
  respectively.  \smallskip
\end{theorem}
\begin{proof}
  Let $x_{\text{f}i}$ be the target state of the $i$-th cluster, and
  let $\| x_{\text{f}i} \|_2 = \alpha_i$. From equations \eqref{eq:
    energy lambda min} and \eqref{eq:decoupled law}, and from the
  definition of $\mc L_2$ gain \cite{SS-IP:05} it follows that
  \begin{align*}
    \| v_i \|_{2,T} &\le \frac{\alpha_i}{\lambda_\text{min}^{1/2} (\mc
      W_{i,T})} , \;\; \| f_{ij} \|_{2,T} \le \frac{\gamma_{ij}
      \alpha_{j}}{ \lambda_\text{min}^{1/2} (\mc W_{j,T}) } .
  \end{align*}
  Moreover, due to the triangle inequality, we have
  \begin{align*}
    \| u_i \|_{2,T} &\le \| v_i \|_{2,T} + \sum_{j \in \mc N_i} \|
    f_{ij} \|_{2,T} \\ &\le \frac{\alpha_i}{\lambda_\text{min}^{1/2}
      (\mc W_{i,T})} + \sum_{j \in \mc N_i} \frac{\gamma_{ij}
      \alpha_{j}}{ \lambda_\text{min}^{1/2} (\mc W_{j,T}) } = \Gamma_i
    \Lambda^{1/2} \alpha,
  \end{align*}
  where $\Gamma_i$ is the $i$-th row of $\Gamma$ defined in
  \eqref{eq:gain2}, and $\alpha$ is the vector of $\alpha_{i}$ with $i
  \in \until{N}$. By using \eqref{eq:gain2} and the fact that $\|
  u_{\mc K}^\text{d} \|_{2,T}^2 = \sum_{i = 1}^N \| u_i \|_{2,T}^2$,
  we obtain
  \begin{align*}
    \| u_{\mc K}^\text{d} \|_{2,T}^2 \le \max_{\| \alpha \| = 1}
    \alpha^\transpose \Lambda^{1/2} \Gamma^\transpose \Gamma
    \Lambda^{1/2} \alpha = \lambda_\text{max} \left( \Lambda^{1/2}
      \Gamma^\transpose \Gamma \Lambda^{1/2} \right) ,
  \end{align*}
  from which the statement follows.
\end{proof}

In Theorem \ref{thm: energy decoupled} we derive a bound on the energy
needed to control a network via our decoupled control law. Theorem
\ref{thm: energy decoupled} has several general consequences which we
now describe. First, due to equation \eqref{eq: energy lambda min}, if
the set $\mc K$ of control nodes includes the boundary nodes of a
network partition $\mc P$, then
\begin{align}\label{eq:lower bound}
  \lambda_\text{min} (\mc W_{\mc K,T}) \ge \frac{1}{\| \Gamma
    \Lambda^{1/2} \|_2^{2} } ,
\end{align}
where $\Lambda$ and $\Gamma$ are the local energy matrix and the $\mc
L_2$ gains matrix for the partition $\mc P$. This bound on the
smallest eigenvalue of the controllability Gramian is novel (see
\cite{HWW-SYM-SCA:96}), and it highlights that the controllability of
a clustered network depends on the controllability of the isolated
clusters via the matrix $\Lambda$, and on their interconnections
strength via the $\mc L_2$ gains matrix $\Gamma$. Second, the control
energy for our decoupled control law does not depend on the
cardinality of the whole network. In fact, notice that
\begin{align}\label{eq: decoupled bound}
  \| \Gamma \Lambda^{1/2} \|_2^2 &\le \| \Gamma \|_2^2 \| \Lambda \|_2
  \le \| \Gamma \|_1 \| \Gamma \|_\infty \|\Lambda\|_\infty ,
\end{align}
and that, independently of the network dimension, $\| \Gamma \|_1$ and
$\| \Gamma \|_\infty$ remain bounded if, for instance, the network
weights and the nodes degrees are bounded. A related example is in
Section \ref{section:example circle}. Third, our decoupled control
strategy is best suited for inherently clustered networks, consisting
of weakly coupled components. Finally, since the energy to control a
network via the decoupled control law depends on local properties of
the network partitions, an appropriate partitioning method may be
developed to optimize the performance of the decoupled control law. To
this aim, we state the following corollary of Theorem \ref{thm: energy
  decoupled}, where we derive a bound on the control energy for our
decoupled control law, which is proportional to the interconnection
strength among clusters. Let $\Delta$ be the \emph{interconnection
  matrix} defined by
\begin{align}\label{eq: interconnection matrix}
  \Delta :=
  \begin{bmatrix}
    1 & \| A_{12} \|_2&\cdots & \| A_{1N} \|_2\\
    \| A_{21} \|_2 & 1&\cdots & \| A_{2N} \|_2\\
    \vdots & \vdots &\ddots & \vdots\\
    \|A_{N1} \|_2&\|A_{N2} \|_2& \cdots & 1
  \end{bmatrix}
  .
\end{align}

\begin{corollary}{\bf \emph{(Bound for network
      partitioning)}}\label{corollary: local bound}
  Let $\gamma_{ij}$ be the $\mc L_2$ gain of the system $(A_j ,B_{\mc
    K_j} , B_{\mc K_i}^\transpose A_{ij})$, and let $\bar
  \lambda_\text{max} =\max \setdef{\lambda_\text{max}(A_i)}{ i \in
    \until{N}}<1$. Then,
  \begin{align*}
    \gamma_{ij} \le \frac{\| A_{ij}\|_2 }{1 - \bar \lambda_\text{max}}
    , \text{ for } j \in \until{N} \setminus \{i\},
  \end{align*}
  and, being $T$ the control horizon,
  \begin{align*}
    \text{E} ( u_{\mc K}^\text{d} , T) \le \frac{ \| \Lambda\|_\infty
      \| \Delta \|_1 \| \Delta \|_\infty }{(1 - \bar \lambda_\text{max} )^2 } ,
  \end{align*}
  where $\Lambda$ is the local energy matrix defined in equation
  \eqref{eq:gain1}, and $\Delta$ is the interconnection matrix defined
  in equation \eqref{eq: interconnection matrix}.
\end{corollary}
\medskip
\begin{proof}
  Recall that $\gamma_{ij}$ equals the $H_\infty$ gain of the transfer
  matrix of the system $(A_j, B_{\mc K_j}, B_{\mc K_i}^\transpose
  A_{ij})$, that is,
  \begin{align*}
    \gamma_{ij} := \left\|B_{\mc K_i}^\transpose
      A_{ij}(zI-A_j)^{-1}B_{\mc K_j}\right\|_{H_\infty} ,
  \end{align*}
  where $\left\|\cdot\right\|_{H_\infty}$ denotes the $H_\infty$ norm
  \cite{SS-IP:05}.  Since the $H_\infty$ norm satisfies the
  submultiplicative property, we have
  \begin{align*}
    \gamma_{ij}\le \left\|B_{\mc K_i}^\transpose \right\|_{H_\infty}
    \left\|A_{ij}\right\|_{H_\infty}
    \left\|(zI-A_j)^{-1}\right\|_{H_\infty} \left\| B_{\mc
        K_j}\right\|_{H_\infty} .
  \end{align*}
  Notice that the $H_\infty$ norm of a constant transfer matrix
  coincides with its induced $2$-norm. Finally we have $\left\|B_{\mc
      K_i}^\transpose \right\|_2 = \left\|B_{\mc K_j} \right\|_2 = 1$,
  and
\begin{align*}
  &\left\|(zI-A_j)^{-1}\right\|_{H_\infty} := \max_{\theta} \;
  \sigma_\text{max} \left(
    (e^{-{\textit{i} \theta}}I - A_j)^{-1} \right)\\
  &=\max_{\theta} \left[\lambda_\text{max} \left((e^{{\textit{i} \theta}}I - A_j )^{-1}(e^{-{\textit{i} \theta}}I - A_j )^{-1}\right)\right]^{1/2}\\
  &=\max_{\theta} \left[\lambda_\text{max} \left(I-2\cos(\theta) A_j+ A_j ^2)^{-1}\right)\right]^{1/2}\\
  &=\frac{1}{ 1 - \lambda_\text{max}(A_j)} \le \frac{1}{ 1 - \bar
    \lambda_\text{max}} ,
  \end{align*}
  from which the first part of the statement follows. The second
  statement follows from \eqref{eq:bound-input-dec} and \eqref{eq:
    decoupled bound} and from the fact that $(1 - \bar
  \lambda_\text{max})\| \Gamma\|_\infty\le \| \Delta\|_\infty$ and $(1
  - \bar \lambda_\text{max})\| \Gamma\|_1\le \| \Delta\|_1$.
\end{proof}

Analogously to equation \eqref{eq:lower bound}, from Corollary
\ref{corollary: local bound} we conclude that, if the set $\mc K$ of
control nodes includes the boundary nodes of a network partition $\mc
P$, then
\begin{align*}
  \lambda_\text{min} (\mc W_{\mc K,T}) \ge \frac{(1 - \bar
    \lambda_\text{max})^2}{ \| \Lambda \|_\infty \| \Delta \|_1 \|
    \Delta \|_\infty } ,
\end{align*}
where $\Lambda$ and $\Delta$ are the local energy matrix and the
interconnection matrix for the partition $\mc P$, respectively, and
$\bar \lambda_\text{max}$ is a bound on the spectral radius of the
clusters of $\mc P$. 

We conclude this part by noting that our results lead to a novel
notion of \emph{network controllability centrality}, a fundamental
concept in network analysis \cite{MEJN:10}, where network nodes are
ranked according to the product of their local controllability degree
and their interconnection strength with neighboring nodes. Our notion
of network controllability centrality is motivated by Corollary
\ref{corollary: local bound}, where the control energy is bounded by
the scaled product of the worst-case control energy of the isolated
clusters $\| \Lambda \|_\infty$ (least controllable cluster), and the
worst-case clusters interconnection strength $\| \Delta \|_1\| \Delta
\|_\infty$ (strongest interconnection strength). A comparison between
controllability centrality and other centrality notions is left as the
subject of future research.

\subsection{An example of network control via decoupled control
  law}\label{section:example circle}
\begin{figure}
    \centering
    \includegraphics[width=1\columnwidth]{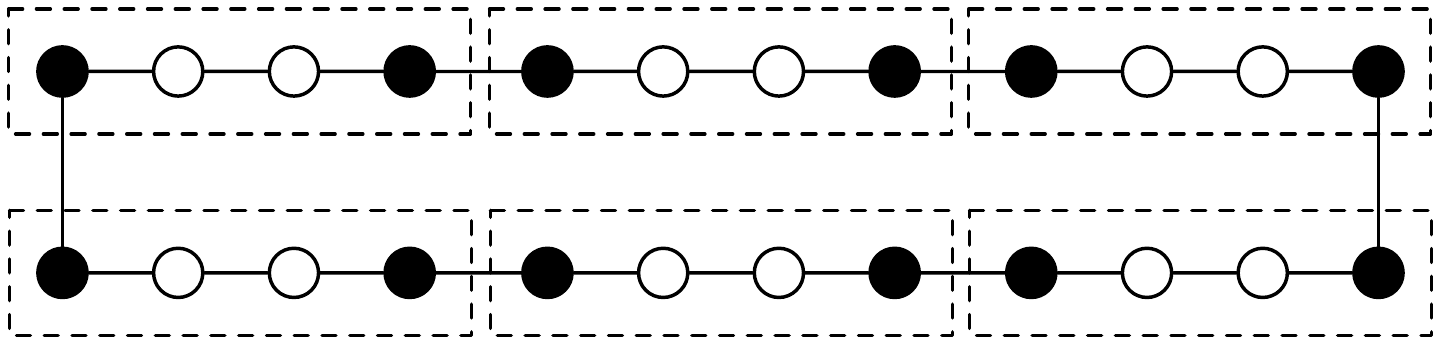}
    \caption{A circulant network with $n = 24$ nodes. The network is
      partitioned into $N = 6$ clusters with $\subscr{n}{b} = 4$ nodes
      each. Controlled nodes are in black.}
    \label{fig:circle}
\end{figure}
\begin{figure}
    \centering
    \includegraphics[width=.49\columnwidth]{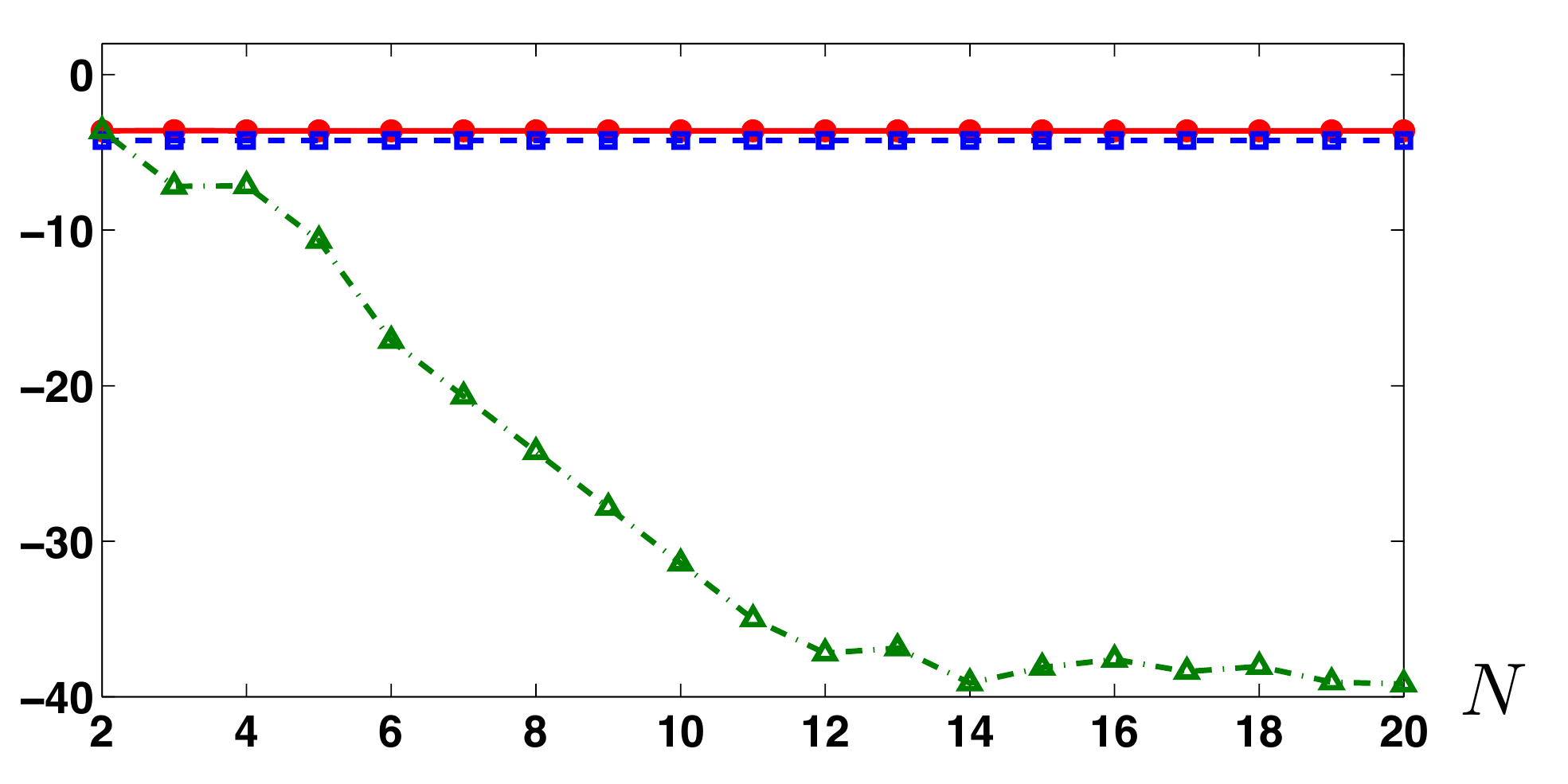}
    \includegraphics[width=.49\columnwidth]{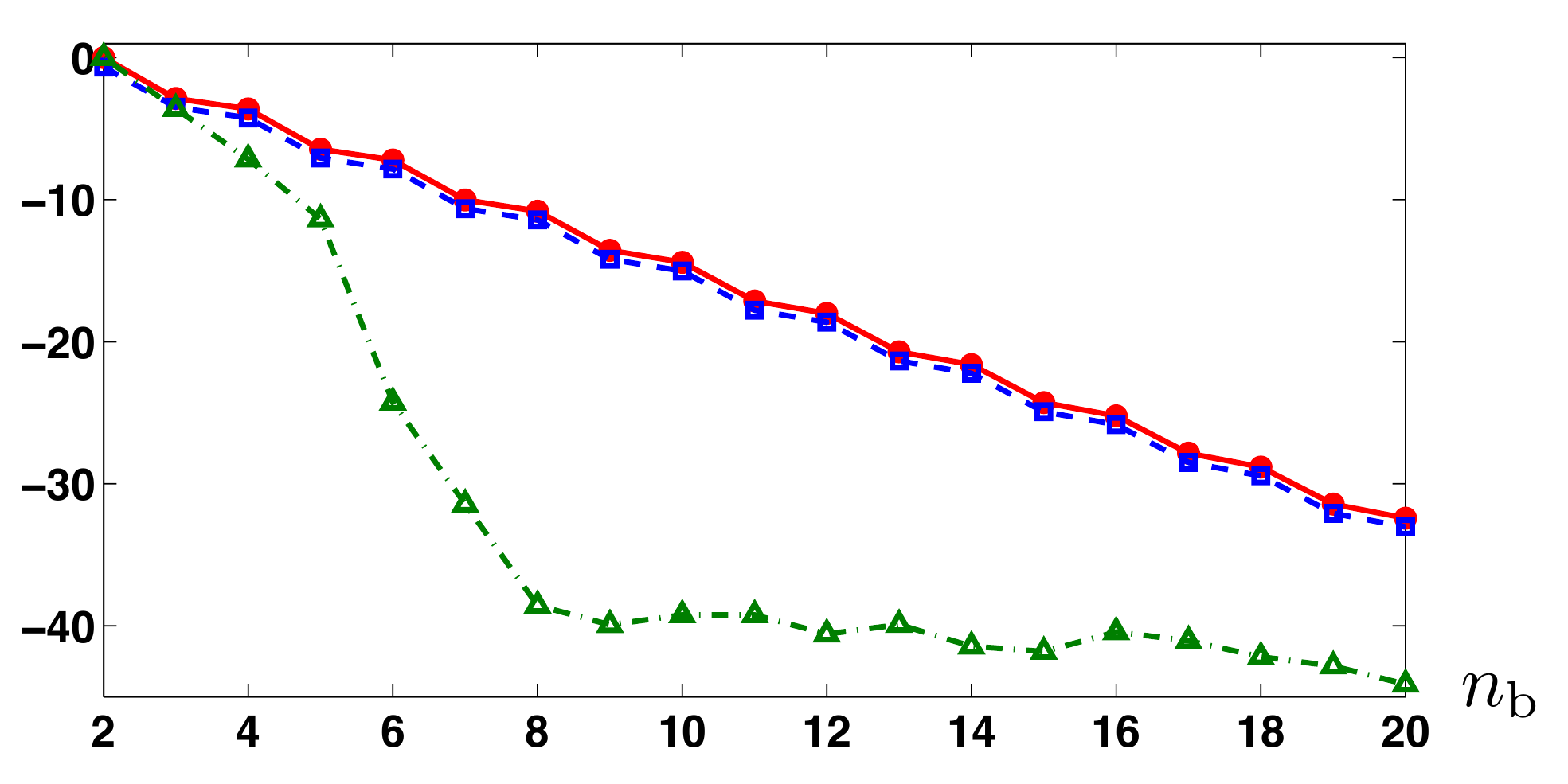}
    \caption{In this figures we study circulant networks partitioned
      as in Section \ref{section:example circle}, and we compare (in a
      logarithmic scale) the performance of our decoupled control law
      against the minimum energy control law. In the left figure we
      maintain constant the number of nodes in each cluster, and we
      report as a function of the number clusters (see Section
      \ref{section:example circle}) (i) the smallest eigenvalue of the
      controllability Gramian with $T = \infty$ and boundary nodes as
      control nodes (solid red), (ii) the bound \eqref{eq:lower bound}
      for the energy performance (see Theorem \ref{thm: energy
        decoupled}) achieved by our decoupled control law (dashed
      blue), and (iii) the smallest eigenvalue of the controllability
      Gramian with $T = \infty$ and control nodes selected randomly
      (dashed-dotted green).  Notice that the energy needed by our
      decoupled control law remains constant when the network
      cardinality grows (the number of control nodes grows as $2N$ and
      that the number of nodes in each cluster remains constant). This
      property is not maintained if the control nodes are chosen
      randomly. In the right figure we report the same quantities as
      in the left figure, while maintaining constant the number of
      clusters and letting the number of nodes in each cluster
      grow. Notice that the smallest eigenvalue of the controllability
      Gramian and our bound \eqref{eq:lower bound} degrade with the
      same rate, while randomly selected control nodes require more
      energy.}
    \label{fig:circle_comparison}
\end{figure}

In this section we demonstrate our technique to control large networks
with an example. Consider a circulant network $\mc G$ with $n =
\subscr{n}{b} N$ nodes, $\subscr{n}{b}, N \in \mathbb{N}$, and
adjacency matrix as in Example \ref{example: bound tight} with $\rho =
0.5$. We partition $\mc G$ into $N$ clusters, so that each cluster
contains $\subscr{n}{b}$ nodes. In particular, we label the nodes in
increasing order,
and for $i\in\{1,\ldots,N\}$ we define the $i$-th cluster to have
vertices $\mc V_i := \{(i-1) \subscr{n}{b} + 1,(i-1) \subscr{n}{b} + 2
, \dots, i \subscr{n}{b}\}$ and control nodes $\mc K_i := \{(i-1)
\subscr{n}{b} + 1, i \subscr{n}{b}\}$.

See Fig. \ref{fig:circle} for an example with $\subscr{n}{b} = 4$ and
$N = 6$. It can be numerically verified\footnote{Due to computational
  complexity, we have solved the maximization problem \eqref{prob:
    optimal placement} for the cases $\subscr{n}{b} = 4$ and $N \in
  \{2,\dots,6\}$.} that the set $\mc K$ of control nodes is optimal,
in the sense that it solves the maximization problem
\begin{align}\label{prob: optimal placement}
  \begin{split}
    &\max_{\mc K \subseteq \until{n}} \; \;\;\lambda_\text{min} (\mc W_{\mc K,\infty}),\\
    &\text{subject to} \;\;\;\;\; |\mc K| = 2 N .
  \end{split}
\end{align}


In Fig. \ref{fig:circle_comparison} we validate Theorem \ref{thm:
  energy decoupled} and equation \eqref{eq:lower bound}. Notice that,
although conservative, our bound \eqref{eq:lower bound} captures the
fact that circulant
networks 
can be driven with constant energy to any (unit norm) target state
independently of the network dimension; this result is compatible with
our analysis in Corollary \ref{thm:upper bound energy} and in Section
\ref{sec: analysis decoupled}.  Moreover, our decoupled control law is
a distributed control law achieving this performance. Finally, it can
be shown that for circulant networks, and in fact for all
$d$-dimensional torus networks, the diagonal entries of
$(I-AA^\transpose)^{-1}$ are all equal to each other. Thus, the
selection of the control nodes for the maximization of the trace of
the controllability Gramian is in this case equivalent to a random
positioning of the control nodes (see the Appendix and \cite[Lemma
3.1]{AC-MM:12}).

\section{Examples of Control of Complex
  Networks}\label{sec:simulations}
The main purpose of this section is to illustrate the effectiveness of
our decoupled control law to control complex networks. To this aim, we
first develop a method to select the control nodes based on network
partitioning, and then compare the performance of the decoupled
control law with alternative control schemes. The design of optimal
partitioning algorithms to minimize the energy of the decoupled
control law, and a thorough comparison with existing partitioning
methods \cite{SF:10} are beyond the scope of this work.

\subsection{Selection of the control nodes}\label{sec: nodes selection}
\begin{algorithm}[t]
   {\footnotesize
    
      \SetKwInOut{Input}{Input}
      \SetKwInOut{Output}{Output}
      \SetKwInOut{Set}{Set}
      \SetKwInOut{Title}{Algorithm}
      \SetKwInOut{Require}{Require}

      \Input{Network $\mc G := (\mc V, \mc E)$, Number of control nodes
        $m$\;}


      \Output{Control nodes $\mc K$\;}

      \BlankLine
      \BlankLine
      
      \nl Define an empty set of control nodes $\mc K := \emptyset$\;

      \smallskip

      \nl Initialize trivial partition $\mc P := \mc V$ with no
      boundary nodes $\subscr{\Psi}{f} := \emptyset$\;

      \smallskip

      \While{$|\mc K| < m$}{
        
        \smallskip

        \nl Select least controllable cluster $\ell = \arg\min
        \setdef{\lambda_\text{min} (\mc W_{i,T})}{ i \in \until{|\mc
            P|}}$ \;

        \smallskip
        
        \nl Compute Fiedler two-partition $\subscr{\mc P}{f}$ of
        $\ell$-th cluster\;
        
        \smallskip

        \nl Compute boundary nodes $\subscr{\Psi}{f}$ of $\subscr{\mc P}{f}$\;

        \smallskip
        
        \nl Update partition $\mc P$ with $\subscr{\mc P}{f}$\;
        
        \smallskip

        \nl Update control nodes with boundary nodes $\mc K = \mc K
        \cup \subscr{\Psi}{f}$\;

      }

      \smallskip

      \nl \lIf{$|\mc K| > m$}{
        Remove boundary nodes of last partition $\mc K = \mc K
        \setminus \subscr{\Psi}{f}$}

      \smallskip

      \nl \lIf{$|\mc K| < m$}{ Add $m - |\mc K|$ control nodes to $\mc
        K$ as in Remark \ref{remark:selection within cluster}}

      \smallskip 

      \nl \Return $\mc K$\;
        
    }
  \caption{\textit{Selection of the control nodes}}
  \label{algo: selection}
\end{algorithm} 

For a connected network $\mc G := (\mc V , \mc E)$ with weighted
adjacency matrix $A$, let $\subscr{\mc P}{f} := \{\mc V_1 , \mc V_2\}$
be the two-partition of $\mc G$ determined by its Fiedler eigenvector
\cite{SF:10,MF:73},\footnote{Let $\subscr{v}{f}$ be the Fiedler
  eigenvector of the network Laplacian matrix. The two-partition
  determined by $\subscr{v}{f}$ is uniquely determined by the sign of
  the entries of $\subscr{v}{f}$ \cite{SF:10}.} and let
$\subscr{\Psi}{f}$ be the boundary nodes of the partition $\subscr{\mc
  P}{f}$.  Our method to select control nodes in a connected network
is described in Algorithm \ref{algo: selection}. Loosely speaking, our
method consists of recursively computing Fielder partitions of
subnetworks of $\mc G$, and selecting the boundary nodes of each
partition as control nodes. Notice that (i) the algorithm repetitively
selects control nodes in the least controllable cluster to improve
local controllability (line $3$), (ii) the set of control nodes
contains the boundary nodes of a network partition (lines $4,5,7$), so
that our decoupled control law can be implemented, and (iii) the set
of control nodes $\mc K$ is increasing throughout the execution of the
algorithm. Consequently, the smallest eigenvalue of the
controllability Gramian is nondecreasing throughout the execution of
the algorithm. In the last part of the algorithm (line $9$) remaining
control nodes are assigned according to a heuristic procedure. Notice
that Algorithm \ref{algo: selection} may return a set of control nodes
from which the network is not controllable. See Section \ref{sec:
  setup decoupling law} for a discussion of this issue.


\begin{remark}{\bf \emph{(Heuristic selection of control
      nodes)}}\label{remark:selection within cluster}
  Different methods can be used to select control nodes in a
  network. Combinatorial methods, heuristic procedures, or random
  selection methods should be employed depending on the network
  dimension and the available computational power. We propose the
  following heuristic method inspired by the notion of \emph{modal
    controllability} \cite{AMAH-AHN:89} to select control nodes within
  each cluster.

  Let $V = [v_{ij}]$ be the matrix of normalized eigenvectors of the
  network adjacency matrix $A$.
  The entry $v_{ij}$ is a measure of the controllability of the mode
  $\lambda_j (A)$ from the control node $i$. In fact, an application
  of the classic PBH test to symmetric matrices shows that $v_{ij} =
  0$ implies that the mode $\lambda_j (A)$ is not controllable from
  node $i$ \cite{TK:80}. By extension, if $v_{ij}$ is small, then the
  $j$-th mode is poorly controllable from node $i$. Let $\phi_i =
  \sum_{j =1}^n (1 - \lambda_j^2 (A)) v_{ij}^2$, and notice that
  $\phi_i$ is a scaled measure of the controllability of all $n$ modes
  $\lambda_1 (A),\dots, \lambda_n (A)$ from the control node $i$. We
  heuristically select the set $\mc K$ of control nodes to maximize
  the smallest controllability parameter $\phi_i$, that is, the set
  $\mc K$ of control nodes is the solution to the maximization problem
  \begin{align}\label{eq:heuristic}
    \begin{split}
      &\max_{\mc K \subseteq \until{n}} \min \; \{\phi_1,\dots,\phi_{k} \} ,\\
      &\text{subject to} \;\;\; |\mc K| = k ,
    \end{split}
  \end{align}
  for a given cardinality $k \in \mathbb{N}$. We remark that our
  heuristic is computationally as hard as computing the eigenvectors
  of each cluster, as the maximization problem \eqref{eq:heuristic}
  can be solved by simply ordering the controllability parameters
  $\phi_i$. \oprocend
\end{remark}


\subsection{Illustrative examples}\label{sec: examples}
In this section we validate our method to control complex networks
with three examples from power networks, social networks, and
epidemics spreading.

\noindent
\textbf{Power network} We consider a network of $n$ generators, and we
describe the dynamics of the $i$-th generator by the linearized swing
equation \cite{PK:94}
\begin{align*}
  m_i \ddot \delta_i + d_i \dot \delta_i = - \sum_{j \in \mc N_i}
  k_{ij} (\delta_i - \delta_j) ,
\end{align*}
where, for the $i$-th generator, $m_i > 0$ and $d_i > 0$ are the
inertia and damping coefficients,
$\map{\delta_i}{\real}{\real_{[0,2\pi]}}$ is the phase angle, and
$k_{ij}$ is the susceptance of the power line $(i,j)$. As in
\cite{FD-FB:09z} we assume that $m_i / d_i \ll 1$, and we approximate
the generator dynamics with a first-order equation. Finally, we
discretize the network by using the Euler method with discretization
accuracy $h$, so that the dynamics of the $i$-th generator read as
\begin{align*}
  \delta_i(t+1) = \delta_i (t) - \frac{h}{d_i} \sum_{j \in \mc N_i}
  k_{ij} (\delta_i (t) - \delta_j (t)) .
\end{align*}

For our numerical study we consider the standard IEEE 118 bus system
with numerical parameters taken from \cite{RDZ-CEM-DG:11}. We assume
that every bus is connected to a generator, and we let the
discretization accuracy be $h = 10^{-7}$. The results of this
numerical study are in Fig. \ref{fig:IEEE118}.

\begin{figure*}[tb]
  \centering \subfigure[IEEE 118 bus system]{
    \includegraphics[width=.64\columnwidth]{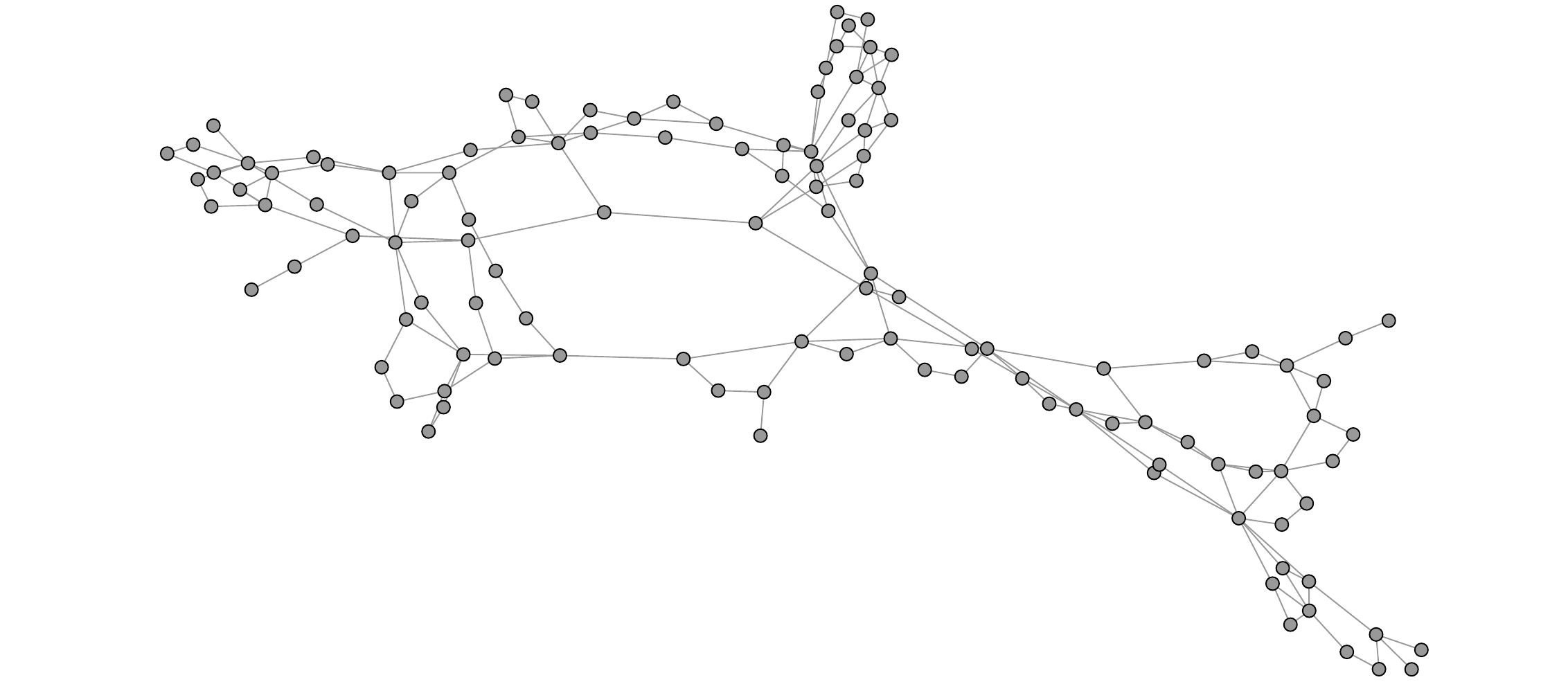}
    \label{fig:IEEE118Layout}
  } \;\,\subfigure[Klavzar bibliography]{
    \includegraphics[width=.64\columnwidth]{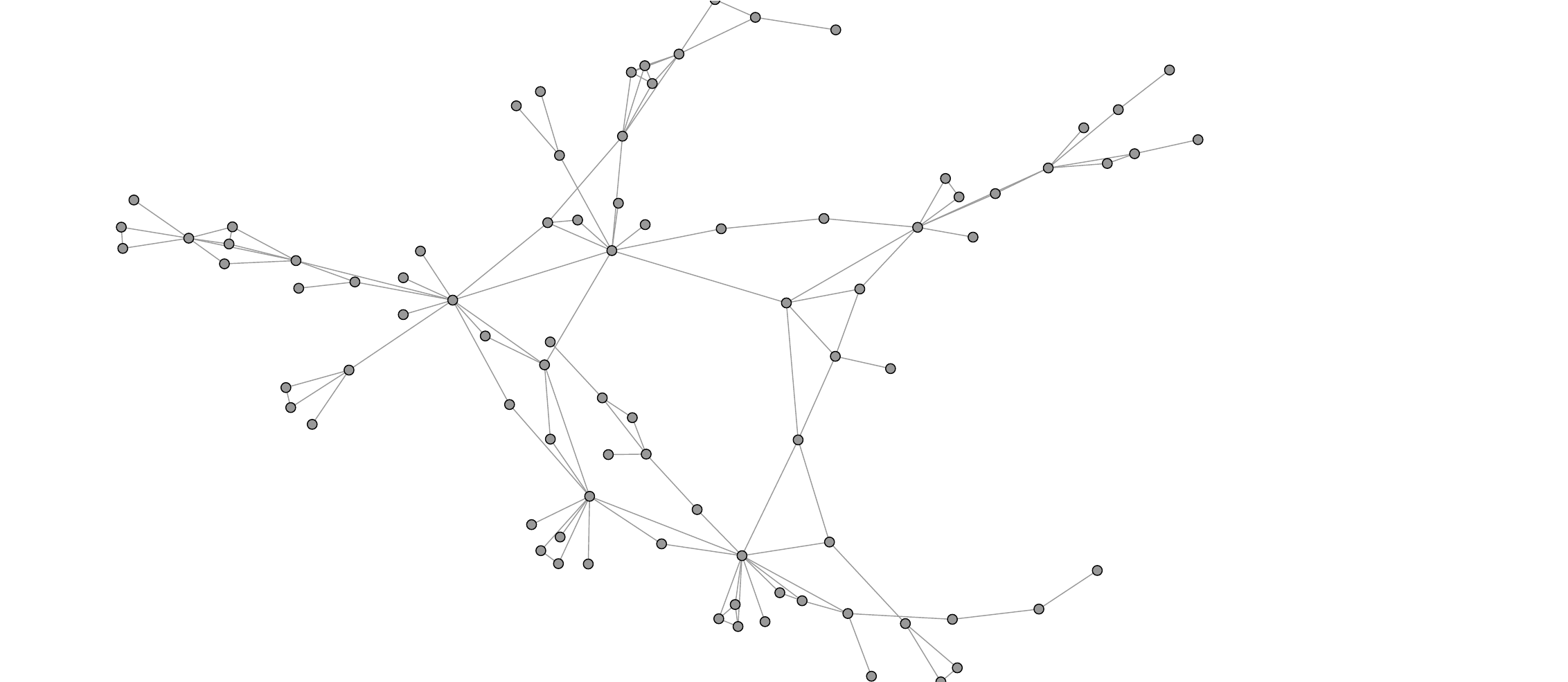}
    \label{fig: Klavzar_Layout}
  }
  \subfigure[Pajek network GD99c]{
    \includegraphics[width=.64\columnwidth]{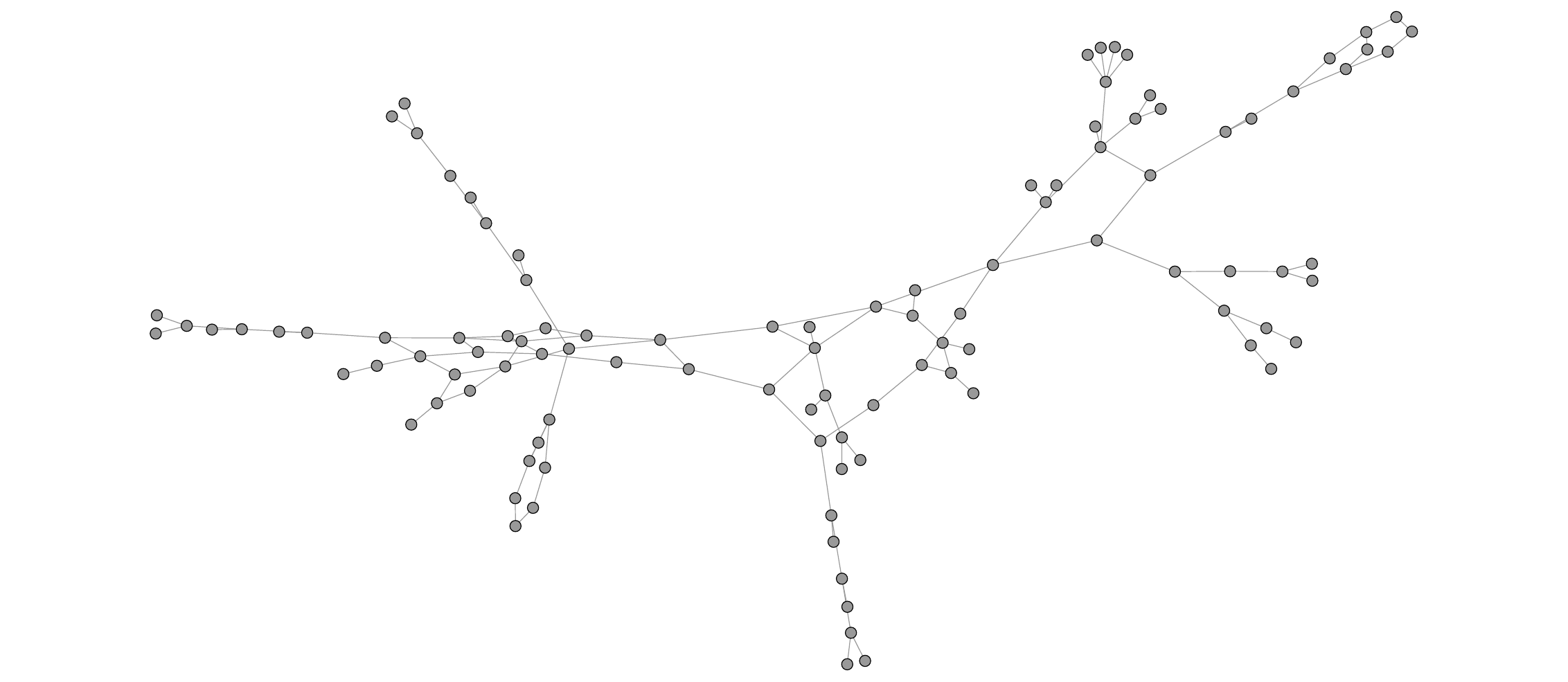}
    \label{fig: epidemics layout}
  }
  \caption[Optional caption for list of figures]{In this figure we
    report a representation of the example networks in Section
    \ref{sec: examples}: Fig. \ref{fig:IEEE118Layout} represents the
    standard IEEE 118 bus system (118 nodes); Fig. \ref{fig:
      Klavzar_Layout} represents the Klavzar bibliography network (86
    nodes); Fig. \ref{fig: epidemics layout} represents the GD99c
    Pajek network (105 nodes). Networks parameters are available at
    \url{http://www.cise.ufl.edu/research/sparse/matrices/}, and their
    layout is obtained via the graph drawing algorithm described in
    \cite{YH:05}.}
  \label{fig:networks}
\end{figure*}

\begin{figure*}[tb]
  \centering \subfigure[IEEE 118 bus system]{
    \includegraphics[width=.65\columnwidth]{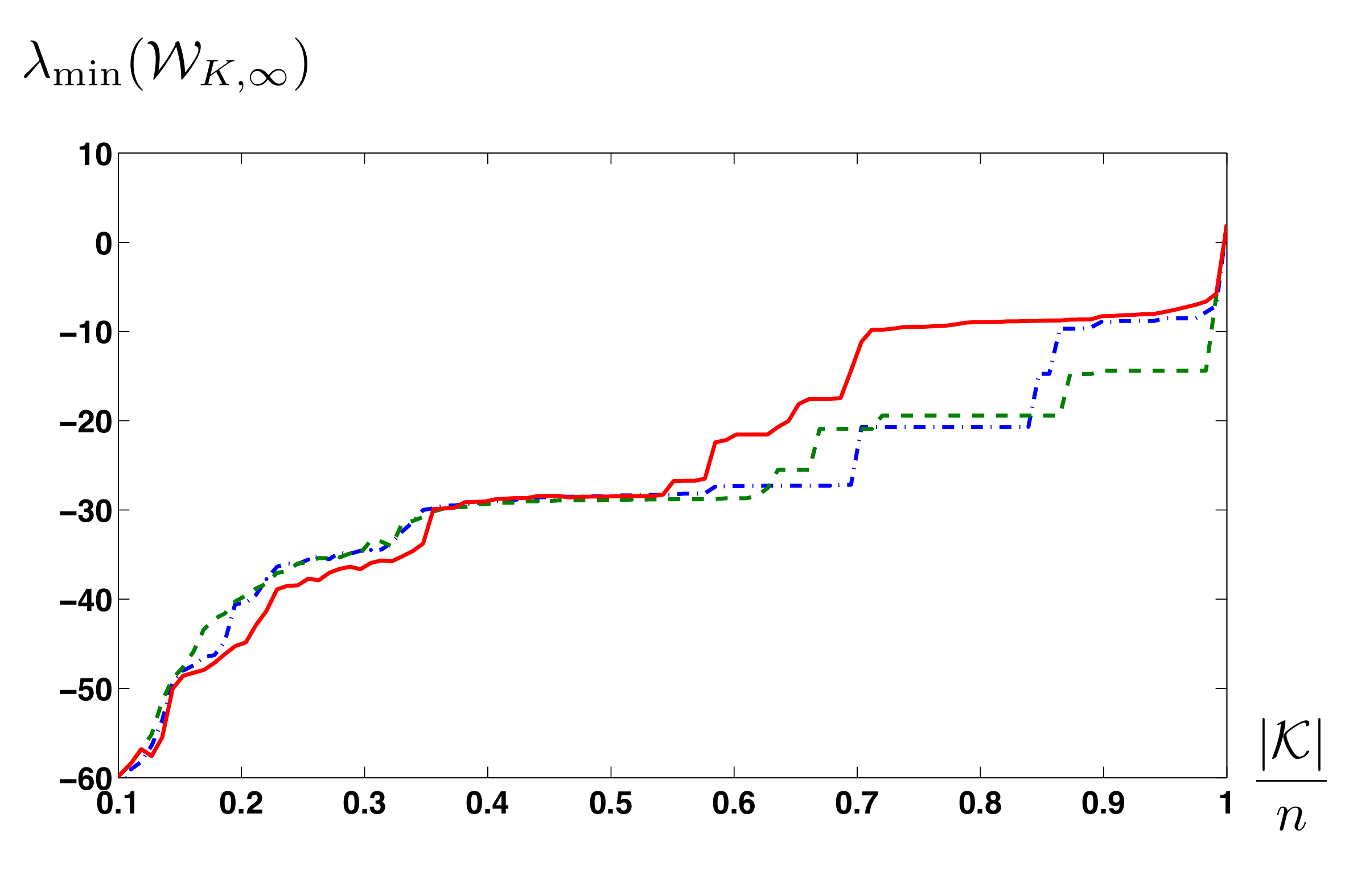}
    \label{fig:IEEE118} }\subfigure[Klavzar bibliography]{
    \includegraphics[width=.65\columnwidth]{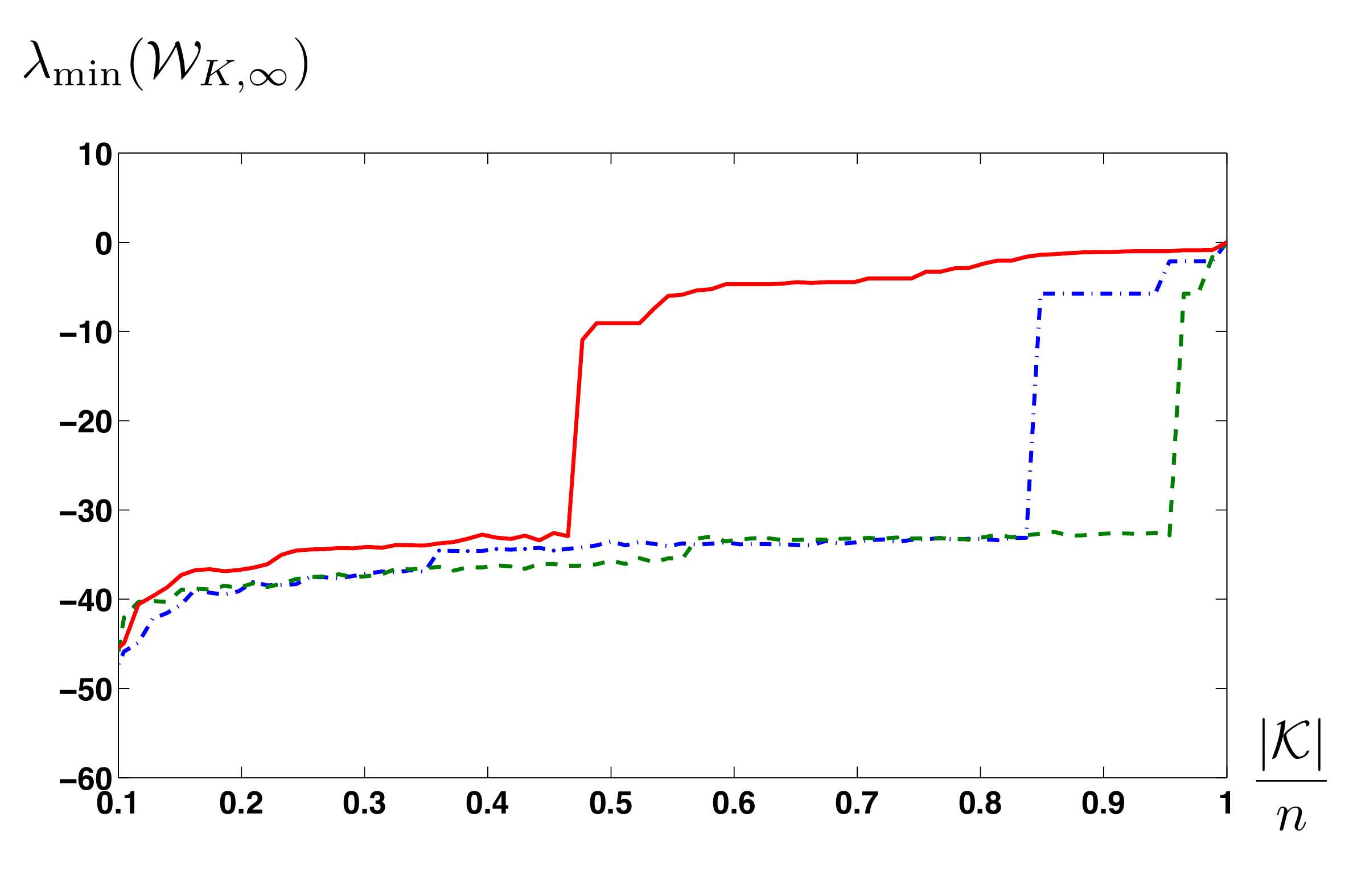}
    \label{fig: klavzar}
  }
  \subfigure[Pajek network GD99c]{
    \includegraphics[width=.65\columnwidth]{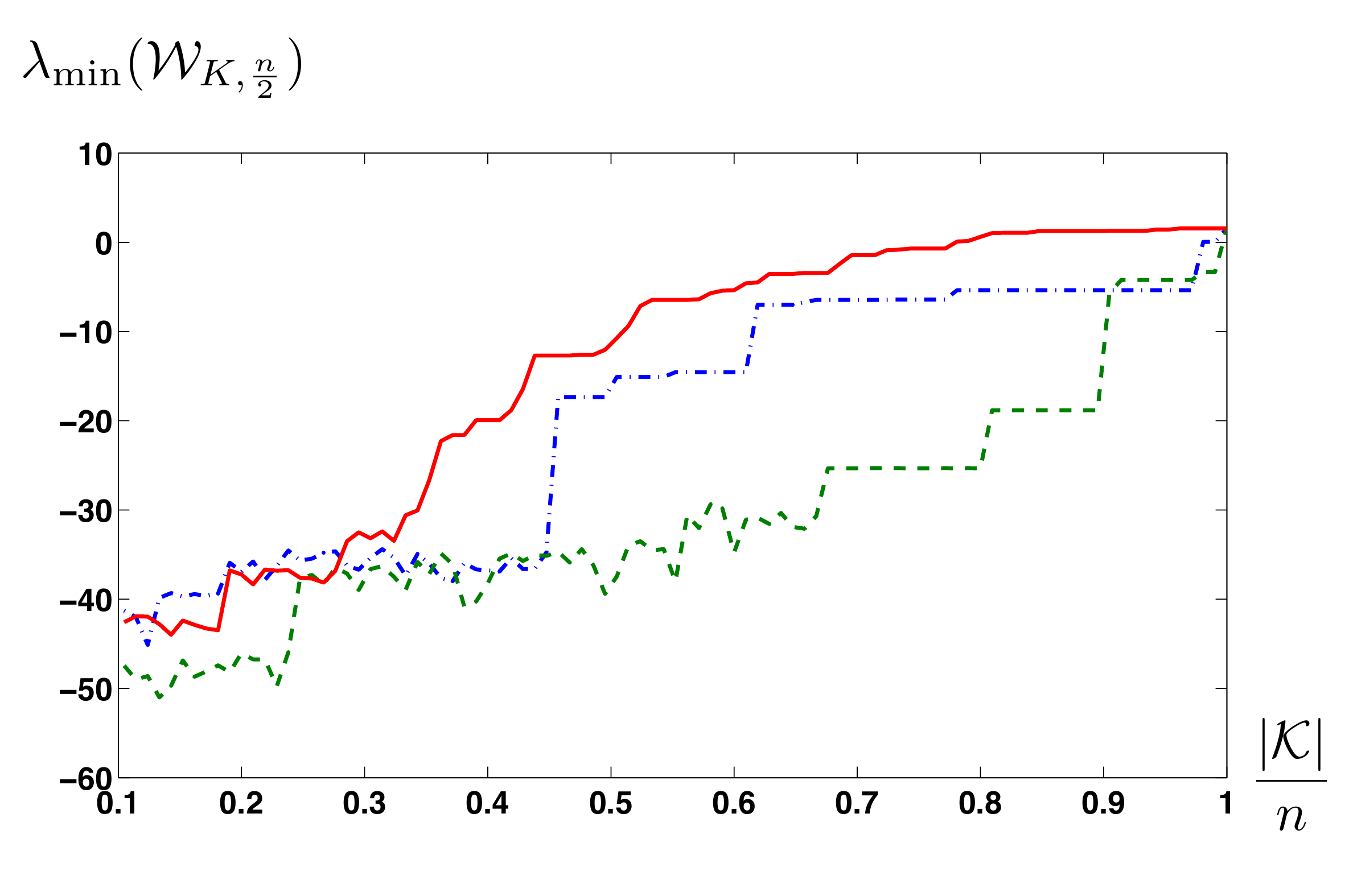}
    \label{fig: epidemics}
  }
  \caption[Optional caption for list of figures]{In this figure we
    compare (in a logarithmic scale) the smallest eigenvalue of the
    controllability Gramian for different choices of the set of
    control nodes. The set of control nodes $\mc K$ is selected
    according to Algorithm \ref{algo: selection} (solid red), trace
    optimization as in Appendix (dashed green), and randomly
    (dashed-dot blue). The cardinality of the control set varies from
    $1$ to $n$. Our decoupled control law algorithm outperforms the
    two counterparts, while being amenable to distributed
    implementation as discussed in Section \ref{sec: setup decoupling
      law}.}
  \label{fig:simulation}
\end{figure*}


\noindent
\textbf{Social network} Inspired by the seminal work \cite{MHDG:74},
the opinion dynamics of a group of individuals forming a network $\mc
G= (\mc V, \mc E)$ can be modeled by the consensus system
\begin{align*}
  x(t + 1) = A x(t) ,
\end{align*}
where $\map{x}{\mathbb{N}}{\real}$ is the vector of the individual
opinions, and the matrix $A = [a_{ij}]$ is row stochastic and
satisfies $a_{ij} = 0$ whenever the edge $(i,j)$ is not in the edge
set $\mc E$. Besides the description of opinion dynamics, consensus
models have found broad applicability in several domains
\cite{FG-LS:10}.

For our numerical study we consider the social network describing the
Klavzar bibliography (see Fig. \ref{fig: Klavzar_Layout}), and we
construct a consensus system by assigning a random nonzero weight to
each edge in the network. The results of this numerical study are in
Fig. \ref{fig: klavzar}.

\begin{remark}{\bf \emph{(Controllability of consensus
      networks)}}\label{remark: controllability consensus}
  Connected consensus networks feature a simple unit eigenvalue
  \cite{FG-LS:10}, 
  so that the controllability Gramian is not defined for the infinite
  control horizon, as the series $\sum_{\tau = 0}^\infty A^\tau B_{\mc
    K} B_{\mc K}^\transpose (A^\transpose)^\tau$ is not convergent. On
  the other hand, it can be shown that the unit eigenvalue is
  controllable at $T = \infty$ by any nonempty set of control nodes
  with zero energy. Then, without loss of generality, the infinite
  horizon controllability Gramian of consensus networks can be defined
  by restricting the dynamics to the subspace orthogonal to the
  consensus space, where the matrix $A$ is Schur stable.  \oprocend
\end{remark}


\noindent
\textbf{Epidemics spreading} The N-intertwined SIS model for the
dynamics of a viral infection over a network with $n$ nodes and
adjacency matrix $A = [a_{ij}]$ reads as \cite{PvM-JO-RK:09}
\begin{align*}
  \dot p_i = -\alpha_i  p_i + (1 - p_i) \beta_i \sum_{j \in \mc N_i}
  a_{ij} p_j , \; i \in \until{n},
\end{align*}
where $\map{p_i}{\real_{\ge 0}}{\real_{[0,1]}}$ is the map describing
the infection probability of node $i$, and $\alpha_i \in \real_{\ge
  0}$, $\beta_i \in \real_{\ge 0}$ are the curing and infection rates
of the $i$-th node. It is known that, for certain values of the ratios
$\alpha_i / \beta_i$, an initial infection $p(0)$ may spread to all
the nodes in the network or converge to zero. 
We consider the simplified model
\begin{align}\label{eq: linear SIS}
 \dot p_i = -\alpha_i  p_i +  \beta_i \sum_{j \in \mc N_i}
  a_{ij} p_j ,
\end{align} 
which is a good approximation of the N-intertwined SIS model at the
initial phase of the epidemics spreading when $p_i$ is small.
We discretize the system
\eqref{eq: linear SIS} as
\begin{align}\label{eq: epidemics}
   p_i (t+1) = (1 - h \alpha_i )  p_i (t) + h \beta_i
  \sum_{j \in \mc N_i} a_{ij}  p_j(t),
\end{align}
where $h \in \real_{> 0}$ is a sufficiently small discretization
parameter, and we study the problem of controlling the spreading of
the infection throughout the network. Notice that an infection can be
controlled for instance by distributing vaccines.

For our numerical study we consider the Pajek social network GD99c
(see Fig. \ref{fig: epidemics layout}), we let $h = 10^{-2}$, and we
select the parameters $\alpha_i$ and $\beta_i$ randomly so that the
network \eqref{eq: epidemics} is unstable. Due to the instability of
the network, we select a finite control horizon of $n/2$ control
steps. The results of this numerical study are in Fig. \ref{fig:
  epidemics}.

From our numerical analysis we draw the following conclusions. First,
the smallest eigenvalue of the controllability Gramian increases
abruptly when the number of control nodes overcomes a certain
threshold, or, equivalently, the control energy decreases abruptly
when the number of control nodes overcomes a certain threshold. This
phenomena is aligned with the \emph{numerical controllability
  transition} identified in \cite{JS-AEM:13} via numerical
simulation. Second, our decoupled control law outperforms the control
strategies dictated by the optimization of the trace of the
controllability Gramian and by random positioning of the control
nodes, while allowing for a distributed and local implementation of
the control law. The difference between the three compared strategies
becomes more evident when the number of control nodes is large. Third
and finally, since our decoupled control law relies on network
partitioning, and computations are performed only on the obtained
subnetworks, it is scalable with the network cardinality and thus
suitable for application to large networks.

We conclude this section with the following consideration. In
Algorithm \ref{algo: selection} we partition each subnetwork by
computing its Fiedler eigenvector. For large networks, this
partitioning scheme may be inefficient, and it may be replaced by a
partitioning scheme with linear complexity, such as the Louvain method
\cite{VDB-JLG-RL-EL:08,JCD-SNY-MB:10}. In this case, our method to
control complex networks has linear complexity, since the decoupled
control law requires only the inversion of local controllability
Gramians whose dimension is independent of the network cardinality. On
the other hand, assuming that the Gauss-Jordan elimination algorithm
is used for the inversion of the controllability Gramian \cite{VS:69},
the computational complexity of the minimum energy control law
\eqref{eq:min energy control} grows at least cubically with the
network cardinality.



\section{Conclusion}\label{sec: conclusion}
In this work we study the problem of controlling complex networks to a
target state. We adopt the smallest eigenvalue of the controllability
Gramian as measure of network controllability, which quantifies the
worst-case control energy. We characterize tradeoffs between the
number of control nodes and the control energy as a function of the
network dynamics. We develop a control strategy with performance
guarantees, consisting of a method to select control nodes based on
network partitioning, and a distributed control law to reach the
target state. Finally, we validate our findings with power systems,
social networks, and epidemics spreading examples.

Important aspects requiring further investigation include (i) the
derivation of tighter bounds for the tradeoff between the number of
control nodes and the control energy, as a function of network
properties, (ii) the study of different controllability measures,
possibly capturing the distributed nature of the problem, and (iii)
the design of an efficient partitioning method to optimize the
performance of our decoupled control law, and (iv) the extension of
our bounds to the design of optimal $\mc H_2$ feedback controllers.

\renewcommand{\theequation}{A-\arabic{equation}}
\setcounter{equation}{0}  
\section*{APPENDIX}  
In this section we derive a closed-form solution to the problem of
selecting control nodes to maximize the trace of the controllability
Gramian, as considered for instance in \cite{THS-JL:13}. To simplify
notation we focus on symmetric networks, although analogous results
hold for asymmetric networks.  Specifically,
we consider the maximization problem
\begin{align}\label{prob: trace}
  \begin{split}
    &\max_{\mc K \subseteq \until{n}} \text{Trace} (\mc W_{\mc K,T}),\\
    &\text{subject to}\;\;\; |\mc K| = m,
  \end{split}
\end{align}
where $m \le n$ and $T \in \mathbb{N}_{\ge 1}$. Notice that
\begin{align*}
  \text{Tr} (\mc W_{\mc K ,T}) &= \text{Tr} \left( \sum_{\tau=0}^{T-1}
    A^\tau B_{\mc K} B_{\mc K}^\transpose A^\tau \right) 
  = \sum_{\tau=0}^{T-1} \text{Tr} \left( B_{\mc K} B_{\mc K}^\transpose
    A^{2\tau} \right) \\ &= \text{Tr} \left( B_{\mc K} B_{\mc
      K}^\transpose \sum_{\tau = 0}^{T-1} A^{2\tau} \right) = \sum_{i \in
    \mc K} \left( \sum_{\tau = 0}^{T-1} A^{2\tau} \right)_{ii} ,
\end{align*}
where we have used that trace is a linear map and is invariant under
cyclic permutations \cite{RAH-CRJ:85}, and where $\left( \sum_{\tau =
    0}^{T-1} A^{2\tau} \right)_{ii}$ denotes the $i$-th diagonal entry
of the matrix $\sum_{\tau = 0}^{T-1} A^{2\tau}$. We conclude that a
solution to the maximization problem \eqref{prob: trace} is the set
$\mc K^*$ containing the indices of the $m$ largest diagonal entries
of $\left( \sum_{\tau = 0}^{T-1} A^{2\tau} \right)$. Notice that, if
$A$ is Schur stable,
then $\sum_{\tau = 0}^\infty A^{2\tau} = (I - A^2)^{-1}$.

\bibliographystyle{IEEEtran}
\bibliography{alias,Main,FB,New}

\end{document}